\documentclass[fontsize=11pt]{article}	 

\usepackage[english]{babel}
\usepackage[utf8]{inputenc}
\usepackage[T1]{fontenc}
\usepackage{amsmath,amsfonts,amsthm} 

\usepackage{proof}
\usepackage{graphicx}
\usepackage{url}
\usepackage{hyperref}
\usepackage{enumerate}
\usepackage{tikz}
\usepackage{stmaryrd}
\usepackage{color}

\usepackage{url}
\usepackage{a4wide}

\title{\bf A Quantitative Study of Pure Parallel Processes\thanks{
This research is supported by the CNRS project {\em ALPACA} (PEPS INS2I 2012) and  A.N.R. project {\em MAGNUM}, ANR 2010-BLAN-0204.}}

\author{Olivier Bodini\thanks{Laboratoire d'Informatique de Paris-Nord,
CNRS UMR 7030 - Institut Galil\'ee - Universit\'e Paris-Nord,
99, avenue Jean-Baptiste Cl\'ement, 93430 Villetaneuse, France. \url{Olivier.Bodini@lipn.univ-paris13.fr}},\ 
Antoine Genitrini\thanks{
Sorbonne Universités, UPMC Univ Paris 06, UMR 7606, LIP6, F-75005, Paris, France.}
\thanks{CNRS, UMR 7606, LIP6, F-75005, Paris, France.\url{{Antoine.Genitrini; Frederic.Peschanski}@lip6.fr}}
\ and 
Fr\'ed\'eric Peschanski$^{\ddag \S }$.\\}

\date{\today} 

\usepackage{appendix}
\usepackage{enumerate}
\usepackage{tikz}
\usepackage[ruled,vlined]{algorithm2e}
\usepackage{stmaryrd}

\newcommand{\C}{\mathcal{C}}

\newcommand{\G}{\mathcal{G}}
\newcommand{\I}{\mathcal{I}}
\newcommand{\M}{\mathcal{M}}

\newcommand{\T}{\mathcal{T}}

\newcommand{\Z}{\mathcal{Z}}

\newcommand{\IN}{\mathbb{N}}
\newcommand{\IP}{\mathbb{P}}

\newcommand{\BigO}{O}  

\newcommand{\contract}[2]{{#1}\lhd{#2}\,}

\newcommand{\fun}[1]{\mathsf{#1}}

\DeclareMathOperator*{\Seq}{\mbox{\sc Seq}}
\DeclareMathOperator*{\Set}{\mbox{\sc Set}}
\DeclareMathOperator*{\MSet}{\mbox{\sc MSet}}

\DeclareMathOperator*{\Shuf}{\mbox{\sc Sem}}

\theoremstyle{plain}
\newtheorem{theorem}{Theorem}
\newtheorem{definition}[theorem]{Definition}

\newtheorem{proposition}[theorem]{Proposition}

\newtheorem{lemma}[theorem]{Lemma}
\newtheorem{corollary}[theorem]{Corollary}
\newtheorem{fact}[theorem]{Fact}

\newtheorem{observation}[theorem]{Observation}
\newtheorem{invariant}[theorem]{Invariant}
\newtheorem{conjecture}[theorem]{Conjecture}

\begin{document}

\label{firstpage}
\maketitle

\begin{abstract}
  In this paper, we study the interleaving -- or pure merge -- operator
  that most often characterizes parallelism in concurrency theory.  This
  operator is a principal cause of the so-called combinatorial
  explosion that makes very hard - at least from the point of view of
  computational complexity - the analysis of process behaviours e.g. by
  model-checking.  The originality of our approach is to
  study this combinatorial explosion phenomenon on average, relying on
  advanced analytic combinatorics techniques.  We study various measures that
  contribute to a better understanding of the process behaviours
  represented as plane rooted trees: the number of runs
  (corresponding to the width of the trees), the expected total size
  of the trees as well as their overall shape.
  Two practical outcomes
  of our quantitative study are also presented: (1) a
  linear-time algorithm to compute the probability of a concurrent run
  prefix, and (2) an efficient algorithm for uniform random sampling
  of concurrent runs. These provide interesting responses to the combinatorial explosion problem.
\end{abstract}

%
%

\section{Introduction}

A significant part of \emph{concurrency theory} is built upon a simple
\emph{interleaving} operator  named the \emph{pure merge} in 
~\cite{Baeten:book90}. The basic underlying idea is that two independent processes running in
parallel, denoted ${P \parallel Q}$, can be faithfully
simulated by the interleaving of their computations. We denote $a.P$
(resp. $b.Q$) a sequential process that first executes an \emph{atomic
  action} $a$ (resp. $b$) and then continue as a process $P$
(resp. $Q$). 

The interleaving law then states\footnote{When one is interested in a finite axiomatization
 of the pure merge operator, a left variant must be introduced, cf.\cite{Baeten:book90} for details.}: 
\[a.P \parallel b.Q = a.(P\parallel b.Q) + b.(a.P \parallel Q),\]

\indent \hfill where $+$ is interpreted as a branching operator. 

\paragraph{}

The pure merge operator is a principal source of
\emph{combinatorial explosion} when analysing concurrent processes,
e.g. by \emph{model checking}~\cite{clarke1999model}.  This issue has
been thoroughly investigated and many approaches have been proposed to
counter the explosion phenomenon, in general based on compression and abstraction/reduction techniques. 
If several decidability and worst-case complexity results are known, to our knowledge the interleaving of process structures
 as computation trees has not been studied extensively from the \emph{average case} point of view.

In analytic combinatorics, the closest related line of work address the \emph{shuffle} of regular languages, generally on disjoint alphabets~\cite{DBLP:journals/dam/FlajoletGT92,MZ08,GDGLOP08,DPRS10}. 
The shuffle on (disjoint) words can be seen as a specific case of the interleaving of processes
 (for processes of the form $(a_1\ldots a_n) \parallel (b_1\ldots b_m)$).
Interestingly, a quite related concept of interleaving of tree structures has been investigated  in algebraic combinatorics~\cite{BFLR11}, and specially in the context of partly commutative algebras~\cite{DHNT11}. 
We see our work has a continuation of this line of works, now focusing on the quantitative and analytic aspects.

Our objective in this work is to
better characterize the \emph{typical shape} of concurrent process
behaviours as computation trees and for this we rely heavily on \emph{analytic
  combinatorics} techniques, indeed on the \emph{symbolic method}.
 One significant outcome of our study is the emergence of a deep connection between concurrent processes and increasing labelling of combinatorial structures. 
 We expect the discovery of similar increasingly labelled structures while we go deeper into concurrency theory. 
We think this work
follows the idea
of investigating \emph{concrete}
problems with advanced analytic tools. In the same spirit, we
emphasize practical applications resulting from such thorough mathematical studies.
In the present case, we develop algorithmic techniques to
analyse probabilistically the process behaviours through \emph{counting}
 and \emph{uniform random generation}.

Our study is organized as follows. In Section~\ref{sec:shuffle} we define the recursive construction of
the interleaved process behaviours from syntactic process trees, and
study the basic structural properties of this construction. In
Section~\ref{sec:widths} we investigate the number of concurrent runs
that satisfy a given process specification. Based on an isomorphism
with \emph{increasing trees} -- that proves particularly fruitful --
we obtain very precise results. We then provide a precise
characterization of what ``exponential growth'' means in the case of
pure parallel processes. We also investigate the case of non-plane trees.
In Section~\ref{sec:shape} we discuss, both theoretically and experimentally,
 the decomposition of semantic trees by level. This culminates with a rather
 precise characterization of the typical shape of process behaviours. 
We then study, in Section~\ref{sec:size}, the expected size of process behaviours.
This typical measure is precisely characterized by a linear recurrence relation
 that we obtain in three distinct ways. While reaching the same conclusion, each of these three proofs 
provide a complementary view of the combinatorial objects under study. Taken together,
 they illustrate the richness and variety of analytic combinatorics techniques
Section~\ref{sec:randgen} is devoted to practical applications resulting from this quantitative study.
First, we describe a simple algorithm to compute the probability of a run prefix in
linear time. As a by-product, we obtain a very efficient way to
calculate the number of \emph{linear extensions} of a \emph{tree-like
  partial order} or \emph{tree-poset}. The second application is an efficient algorithm 
for the uniform random sampling of concurrent runs. These algorithms work directly on the syntax
trees of process without requiring the explicit construction of their behaviour, thus avoiding the combinatorial
 explosion issue.

This paper is an updated and extended version of~\cite{BGP12}. It contains
 new material, especially the study
 of the typical shape of process behaviours in Section~\ref{sec:shape}.  The more complex
 setting of non-plane trees is also discussed. Appendix~\ref{sec:randmset} was added to discuss the weighted 
random sampling in dynamic multisets. The proofs in this extended version are also more detailed.

\section{A tree model for process semantics}\label{sec:shuffle}

\begin{figure}
\begin{center}
\begin{tabular}{ccc}
\begin{tikzpicture}[node distance=24pt]
\node (a) {$a$};
\node[below of=a] (b) {$b$};
\draw (a) -- (b);
\node[below left of=b] (c) {$c$};
\draw (b) -- (c);
\node[below right of=b] (d) {$d$};
\draw (b) -- (d);
\node[below left of=d] (e) {$e$};
\draw (d) -- (e);
\node[below right of=d] (f) {$f$};
\draw (d) -- (f);
\node[below of=f, node distance=60pt] (ff) {};
\end{tikzpicture}
&
\hspace{40pt}
&
\begin{tikzpicture}[node distance=24pt]
\node(a) {$a$};
\node[below of=a] (b) {$b$};
\draw (a) -- (b);
\node[below of=b] (bb) {};
\node[left of=bb, node distance=60pt] (c1) {$c$};
\draw (b) -- (c1);
\node[below of=c1] (d1) {$d$};
\draw (c1) -- (d1);
\node[below left of=d1] (e1) {$e$};
\draw (d1) -- (e1);
\node[below right of=d1] (f1) {$f$};
\draw (d1) -- (f1);
\node[below of=e1] (f1') {$f$};
\draw (e1) -- (f1');
\node[below of=f1] (e1') {$e$};
\draw (f1) -- (e1');
\node[right of=bb, node distance=60pt] (d2) {$d$};
\draw (b) -- (d2);
\node[below of=d2] (e2) {$e$};
\draw (d2) -- (e2);
\node[below left of=e2] (c5) {$c$};
\draw (e2) -- (c5);
\node[below of=c5] (f5') {$f$};
\draw (c5) -- (f5');
\node[below right of=e2] (f5) {$f$};
\draw (e2) -- (f5);
\node[below of=f5] (c5') {$c$};
\draw (f5) -- (c5');
\node[left of=e2, node distance=48pt] (c2) {$c$};
\draw (d2) -- (c2);
\node[below left of=c2] (e3) {$e$};
\draw (c2) -- (e3);
\node[below of=e3] (f3') {$f$};
\draw (e3) -- (f3');
\node[below right of=c2] (f3) {$f$};
\draw (c2) -- (f3);
\node[below of=f3] (e3') {$e$};
\draw (f3) -- (e3');
\node[right of=e2, node distance=48pt] (f2) {$f$};
\draw (d2) -- (f2);
\node[below left of=f2] (c4) {$c$};
\draw (f2) -- (c4);
\node[below of=c4] (e4') {$e$};
\draw (c4) -- (e4');
\node[below right of=f2] (e4) {$e$};
\draw (f2) -- (e4);
\node[below of=e4] (c4') {$c$};
\draw (e4) -- (c4');
\end{tikzpicture}
\end{tabular}
\end{center}
\caption{\label{fig:task-and-shuffle} A syntax tree (left) and the corresponding semantic tree (right)}
\end{figure}
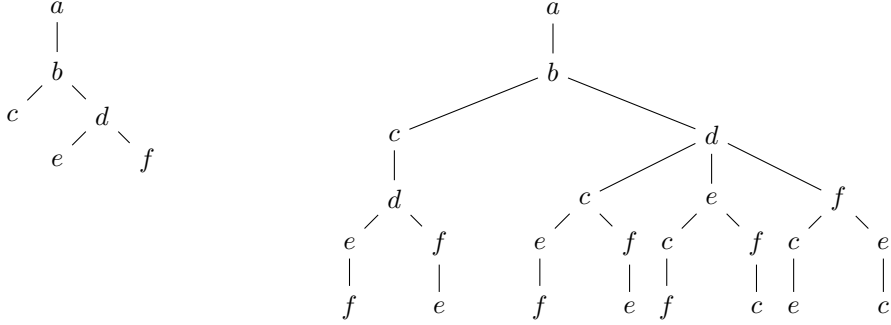

As a starting point, we recast our problematic in combinatorial terms. The idea is to relate the
 \emph{syntactic domain} of process specifications to the \emph{semantic domain} (or model) of
 process behaviours.

\subsection{Syntax trees}

The grammar we adopt for pure parallel
 processes is very simple. The set of process specifications is the least set satisfying:
\begin{itemize}
\item an atomic action, denoted $a,b, \ldots$ is a process,
\item the prefixing $a.P$ of an action $a$ and a process $P$ is a process, and, more precisely, a prefixed process,
\item the composition $P_1 \parallel \ldots \parallel P_n$ of a finite number of actions or prefixed processes is a process.
\end{itemize}

Let us first remark that this grammar takes the $\parallel$ operators as associative operators
and thus two of them cannot appear consecutively.
Moreover, in the rest of the paper we will concentrate on \emph{prefixed processes}. This choice does not deplete
the results thanks to the  bijection between prefixed processes of size $n$ and processes without prefixed action of size $n-1$.

An example of a valid specification is: \[a.b.(c \parallel d.(e \parallel f))\] which can be faithfully
 represented by a tree, namely a \emph{syntax tree}, as depicted on the lefthand side of Figure~\ref{fig:task-and-shuffle}.
Such a tree can be read as a set of \emph{precedence constraints} between atomic actions.
Under these lights the action $a$ at the root must be executed first and then $b$.
There is no relation between $c$ and $d$ -- they are said independent -- and $e,f$ may only happen after~$d$.

In combinatorial terms we adopt the classical specification for \emph{plane rooted trees} to represent the syntactic domain.
The \emph{size} of a tree is its total number of nodes.
Note that we do not keep the names of actions in the process trees since they play no r\^ole for the pure merge operator.

\begin{definition}\label{def:catalan}
The specification $\C = \Z \times \Seq(\C)$ represents the combinatorial class of plane rooted trees. 
\end{definition}

As a basic recall of analytic combinatorics and statement of our conventions, we remind that for such a
combinatorial class $\C$, we define its counting sequence $C_n$
consisting of the number of objects of $\C$ of size $n$. This sequence is linked
to a formal power series $C(z)$ such that $C(z) = \sum_{n\geq 0} C_n z^n$.
We denote by $[z^n]C(z)=C_n$ the $n$-th coefficient of $C(z)$.
Analogous writing conventions will be used for all combinatorial classes in this paper.

We remind the reader that in the case of class $\C$ the sequence $C_n$ corresponds to the
\emph{Catalan numbers} (indeed, shifted by one).  For further reference, we give
the generating function of $\C$ and
the asymptotic approximations of the Catalan numbers
(obtained by the Stirling formula approximation of $n!$ as in e.g.~\cite[P. 267]{Comtet74}):

\begin{fact}\label{fact:Catalan}
$C(z) = \frac12 - \frac{\sqrt{1-4z}}2$ and 
$C_{n} = \frac{4^{n-1}}{\sqrt{\pi n^3}} \left( 1 + \frac{3}{8n} + \frac{25}{128n^2} + \frac{105}{1024n^3} + \frac{1659}{32768n^4} + \BigO\left(\frac1{n^5}\right) \right).$
\end{fact}

\subsection{Semantic trees}

The semantic domain we study is much less classical than syntax trees, although it is \emph{still} composed of
plane rooted trees. An example of a \emph{semantic tree} is depicted on the righthand side of Figure~\ref{fig:task-and-shuffle}.
This tree represents all the possible executions -- or \emph{runs} -- that may be observed for the process specified
 on the left. More precisely each branch of the semantic tree, e.g. $\langle a,b,c,d,e,f\rangle$, is a concurrent run (or admissible computation)
 of the process, and all the branches share their common prefix. In the literature such structures are also called \emph{computation trees}~\cite{DBLP:journals/toplas/ClarkeES86}.


For a tree $T$ and one of its node $v$, the \emph{sub-tree} roooted in $v$ from $T$
is the tree of all the descendants of $v$ in $T$.
To describe the recursive construction of the semantic trees, we use an elementary operation of \emph{child contraction}.

\begin{definition}\label{def:child-contraction}
Let $T$ be a plane rooted tree and $v_1,\ldots,v_r$ be the root-labels of the children of the root.
For $i\in\{1,\dots,r\}$, the $i$\textbf{-contraction} of $T$
is the plane tree with root $v_i$ and whose children are, from left to right,
$T(v_1),\ldots,T(v_{i-1}),T(v_{i_1}),\ldots,T(v_{i_m}), T(v_{i+1}),\ldots,T(v_r)$ 
(where $T(\nu)$ denotes the sub-tree whose root is $\nu$ and
$v_{i_1},\ldots,v_{i_m}$ are the root-labels of the children of~$T(v_i)$).
We denote by $\contract{T}{\,i}$ the $i$-contraction of $T$.
\end{definition}

\begin{tabular}{llll}
For example, if $T$ is &
\begin{tikzpicture}[baseline]
\node (a) {$a$};
\node[below of=a] (c) {$c$};
\draw (a) -- (c);
\node[left of=c, node distance=40pt] (b) {$b$};
\draw (a) -- (b);
\node[right of=c, node distance=40pt] (d) {$d$};
\draw (a) -- (d);
\node[below left of=c] (e) {$e$};
\draw (c) -- (e);
\node[below right of=c] (f) {$f$};
\draw (c) -- (f);
\node[below of=b] (b') {};
\draw[dotted] (b) -- (b');
\node[below of=d] (d') {};
\draw[dotted] (d) -- (d');
\node[below of=e] (e') {};
\draw[dotted] (e) -- (e');
\node[below of=f] (f') {};
\draw[dotted] (f) -- (f');
\end{tikzpicture} &
then $\contract{T}{2}$ is
\begin{tikzpicture}[baseline]
\node (c) {$c$};
\node[below left of=c] (e) {$e$};
\draw (c) -- (e);
\node[below right of=c] (f) {$f$};
\draw (c) -- (f);
\node[left of=e] (b) {$b$};
\draw (c) -- (b);
\node[right of=f] (d) {$d$};
\draw (c) -- (d);
\node[below of=b] (b') {};
\draw[dotted] (b) -- (b');
\node[below of=d] (d') {};
\draw[dotted] (d) -- (d');
\node[below of=e] (e') {};
\draw[dotted] (e) -- (e');
\node[below of=f] (f') {};
\draw[dotted] (f) -- (f');
\end{tikzpicture} &
\end{tabular}

Note that the root (here $a$) is replaced by the label of the root of the $i$-th child (here $c$).
Now, the interleaving operation follows a straightforward recursive scheme.

\begin{definition}\label{def:shuffle}
Let $T$ be a process tree, then its \textbf{semantic tree} $\Shuf(T)$ is defined inductively as follows:
\begin{itemize}
\item if $T$ is a leaf, then $\Shuf(T)$ is $T$,
\item if $T$ has root $t$ and $r$ children ($r\in \IN\setminus\{0\}$), then
 $\Shuf(T)$ is the plane tree with root $t$ and children, from left to right,
 $\Shuf(\contract{T}{1}),\ldots, \Shuf(\contract{T}{r})$.
\end{itemize}
\end{definition}

\begin{figure}
\begin{center}
\input{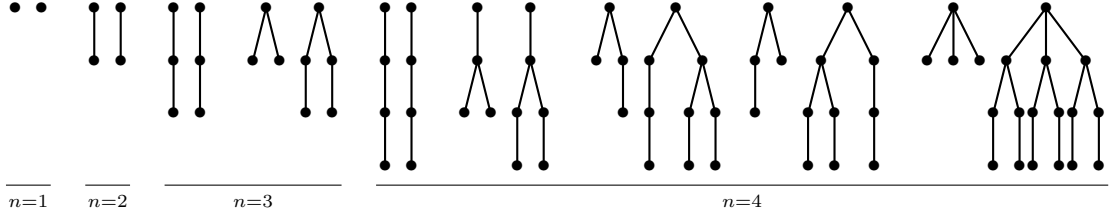}
\end{center}
\caption{\label{fig:enum} Enumerating behaviours (semantic trees) from process specifications (syntax trees).}
\end{figure}

The mapping between the syntax trees on the one side, and the semantic trees on the other side is trivially one-to-one.
Figure~\ref{fig:enum} depicts the enumeration of the first syntax trees  (by size $n$) together with the corresponding semantic tree.

We note that the semantic trees are
\emph{balanced} (i.e. all their leaves belong to the same level),
and even more importantly and that their height is $n-1$,
where $n$ is the size of the associated process tree.  This
is obvious since each branch of a semantic tree corresponds to a
complete traversal of the syntax tree. Thus there are as many semantic
trees of height $n-1$ as there are trees of size $n$ (as counted by $C_n$ above).

A further basic observation is that the contraction operator (cf. Definition~\ref{def:child-contraction})
ensures that the number of nodes at a given level of a semantic tree is lower than the number of
nodes at the next level. Thus, the width of the semantic tree corresponds to its number of leaves.



The following observation bears witness to the high level of redundancy exhibited by semantic trees.

\begin{proposition}\label{prop:branch}
The knowledge of a single branch of a semantic tree is sufficient to recover the corresponding syntax tree.
\end{proposition}

To go slightly further into the details, we may indeed exhibit a familiy of inverse functions from
singled-out semantic tree branches to syntax trees. 
These inverse functions exploit the concept of a \emph{degree-sequence}, defined as follows.
\begin{definition}
A \textbf{degree-sequence} $(u_p)_{p\in\{1,\dots, n\}}$ is a sequence of non-negative integers
of length $n$ that satisfies:
\[u_1 > 0;\hspace{0.8cm}
\forall p>1,\ u_p \geq u_{p-1}-1;
\hspace{0.8cm} 
u_n \text{ is the single term equal to 0.}
\]
\end{definition}

The degrees of the nodes from the root to a leaf in any branch of a semantic tree is a degree-sequence.

\begin{proposition}
Let $(u_p)$ be a degree-sequence of length $n$ that is linked to the leftmost branch of a semantic
tree $S$. Let us define the new sequence $(v_p)$ such that:
\[v_1 = u_1;\hspace{1cm}
\forall p>1,\ v_p = u_p - u_{p-1}+1.\]
We build a tree $T$ of size $n$ such that
the sequence $(v_n)$ corresponds to the degrees of each node of
the tree, ordered by the prefix traversal. The semantic image of $T$ is the tree $S$.
\end{proposition}

An important remark is that we only considered the leftmost branch of the semantic tree to construct
the corresponding degree-sequence, from which we recover the initial syntax tree.
It is interesting to note that the leftmost branch of the semantic tree encodes a {\L}ukasiewicz word which is directly related
to the degree-sequence of the prefix traversal~\cite[p. 74--75]{FS09}.

We can show, in fact, that the initial tree can be recovered by considering \emph{any} of the branches of
its semantic tree, not just the leftmost one. Each branch corresponds to a degree sequence visiting the nodes of the initial tree by a
 specific traversal. For example, if the leftmost branch encodes the prefix traversal;
the rightmost branch enumerates its mirror: the postfix traversal.
Last but not least, the set of degree-sequences of length $n$ is only of
cardinality~$C_n$, so the semantic trees are highly \emph{symmetrical}  
in that many branches must be defined by the same degree-sequence.

\section{Enumeration of concurrent runs}\label{sec:widths}

Our quantitative study begins by measuring
the number of concurrent runs of a process encoded as a syntax tree $T$. This measure in fact corresponds to the number of leaves -- and thus the width --
 of the semantic tree $\Shuf(T)$. Given the exponential nature of the merge operator, measuring efficiently the \emph{dimensions} of the concurrent systems under study is of a great practical interest.
In a second step, we quantify precisely the exponential
 growth of the semantic trees, which provides a refined interpretation of the so-called combinatorial explosion phenomenon.
 Finally, we study the impact of characterizing commutativity for the merge operator.
As a particularly notable fact, this section reveals a deep connection between increasingly
 labelled structures and concurrency theory.

\subsection{An isomorphism with increasing trees}

Our study begins by a simple observation that connects the pure merge operator
 to the set of \emph{linear extensions} of tree-like partial orders or \emph{tree-posets}~\cite{Atkinson90}.

\begin{definition}
Let $T$ be the syntax tree of a process, and $A$ its set of actions. We define the poset $(A,\prec)$ such that
 $a \prec b$ iff  $a$ is the label of a node that is the parent of a node with label $b$ in $T$.
The linear extensions of $(A,\prec)$ is the set of all the strict orderings $(A,<)$ 
 that respect the partial ordering.
\end{definition}

For example, the syntax tree depicted on the left of Figure~\ref{fig:task-and-shuffle}
 is interpreted as the partial order $a\prec b;b \prec c;b \prec d;d \prec e;d \prec f$.

\begin{proposition}
Let $T$ be the syntax tree of process, $(A,\prec)$ the associated tree-poset. Then:
\begin{itemize}
\item Each branch of $\Shuf(T)$ encodes a distinct strict ordering of $A$ that respects $(A,\prec)$,
\item If a strict ordering $(A,<)$ respects $(A,\prec)$ then it is encoded by a given branch of $\Shuf(T)$.
\end{itemize}
\end{proposition}

This observation is quite trivial, and can be justified by the fact that each branch of $\Shuf(T)$ encodes
 a distinct traversal of $T$. For  example in Figure~\ref{fig:task-and-shuffle}  the leftmost branch of the
semantic tree is the linear extension $a<b<c<d<e<f$, that indeed fulfills the tree-poset.

Under these new order-theoretic lights, we can exhibit a deep connection between the number of concurrent
 runs of a syntax tree $T$ and the number of ways to label it in a strictly increasing way.
Indeed, as already observed in~\cite{DMTCSPROCdmAR0168}, the linear extensions of tree-posets are in one-to-one correspondence with \emph{increasing trees}~\cite{DBLP:conf/caap/BergeronFS92,Drmota09}.

\begin{definition}
An \textbf{increasing tree} is a labelled plane rooted tree such that the sequence of labels
along any branch starting at the root is increasing.
\end{definition}
For example, to label the tree of Figure~\ref{fig:task-and-shuffle}, $a$ would take the label~1,
$b$ then takes the label~2. Then the label of $c$ must belong to $\{3,4,5,6\}$, which would
 then it induces constraints on the other nodes. Finally, only $8$ labelled trees are increasing trees among the $6!$ possible
unconstrained labellings.

Increasing plane rooted trees satisfy the following specification
(using the classical boxed product$^\text{\qed} \star$ see~\cite[p.~139]{FS09} for details):
\[\G = \Z^\text{\qed} \star \Seq(\G).\]

It is easy to obtain the coefficients of the associated exponential generating function $G(z)$  (e.g. from~\cite{DBLP:conf/caap/BergeronFS92}):

\begin{fact} \label{fact:number:increasing}
The number of increasing plane rooted trees of size $n$ is 
\[n!\cdot [z^n]G(z) = 1\cdot3\cdots (2n-3) = \frac{(2n-2)!}{2^{n-1} (n-1)!}.\]
\end{fact}

From this we obtain our first significant measure.

\begin{theorem}\label{theo:MeanNbSchedulings}
The mean number of concurrent runs induced by syntax trees of size~$n$ is:
\[\bar{W}_n = \frac{n!}{2^{n-1}}
\sim_{n\rightarrow \infty} 2\sqrt{2\pi n}\left(\frac{n}{2e}\right)^{n}.\]
\end{theorem}

This result is obtained from Fact~\ref{fact:number:increasing} by taking the average number
of increasing trees of size $n$, and the asymptotics is based on Stirling's formula~\cite[p.~37]{FS09}.

A further information that will prove particularly useful is the number of increasing
labellings for a given tree. This can be obtained by the famous \emph{hook-length} formula~\cite[p.~67]{Knuth98}:

\begin{fact}\label{fact:hook}
The number $\ell_T$ of increasing trees built on a plane rooted tree $T$ is:
\[\ell_T = \frac{|T| !}{\prod_{S \text{ sub-tree of } T} |S|},\]
\indent \hfill where $|\cdot|$ corresponds to the size measure.
\end{fact}

\begin{corollary}\label{algo:hook-length}
The number of concurrent runs of a syntax tree $T$ is the number~$\ell_T$.
\end{corollary}

We remark that the hook-length gives us ``for free'' a direct algorithm to compute the number of linear
extensions of a tree-poset in linear time. This is clearly an improvement if compared to related algorithms, e.g.~\cite{Atkinson90}. 
In Section~\ref{sec:randgen} we discuss a slightly more general and more efficient 
 algorithm that proves quite useful.

\subsection{Analysis of growth}\label{sec:growth}

To analyse quantitatively the growth between the processes and their behaviours,
we measure the average number of concurrent runs
induced by large syntax trees of size $n$. The arithmetic mean given in Theorem~\ref{theo:MeanNbSchedulings}
is the usual way to measure in average. 
Nevertheless, a small number of compact syntax trees (such as a root followed by (n-1) sons) produces
a huge number of runs and unbalance the mean. So, a natural way to avoid such bias is to compute
the geometric mean which is less sensitive to extremal data.
This subsection is devoted to prove the following theorem about the geometric mean number of concurrent runs.

\begin{theorem}\label{theo:growth}
The geometric mean number of concurrent runs built on process trees of size~$n$ satisfies:
\[\bar{\Gamma}_n =  \prod_{k=2}^{n-1} k^{1 - \frac{n+1-k}{2}\frac{C_{k}C_{n-k+1}}{C_{n}}}
 \sim_{n\rightarrow\infty} \sqrt{2\pi} \frac{e^{\sqrt{\pi n} + L(1/4)}}{n} \left(\frac{n}{e^{1+2L(1/4)}}\right)^{n},\]
\indent \hfill where\footnote{For approximate constants, the exact digits are written in bold type.}
 $\textstyle{L\left(1/4\right) = \sum_{n>1} \log n \cdot C_{n} \cdot 4^{-n} \approx {\bf 0.5790439}217 \pm 5\cdot10^{-9}}$.
\end{theorem}
This growth appears to us as less important than what we conjectured with the arithmetic mean,
although it is still very large. For both means, the result is indeed quite far from the upper bound $(n-1)!$.


\begin{proof}
First we need to obtain a recurrence formula based on the hook length formula. Let us give the following observation:
\[\prod_{S \text{ sub-tree of } T} |S| = |T| \cdot
 \prod_{R \text{ child of the root of } T} \left(\prod_{S \text{ sub-tree of } R} |S|\right).\]
Now, by Fact~\ref{fact:hook} we deduce the next recursive equation:
\[\ell_T = (|T|-1)! \left( \prod_{R \text{ child of the root of } T} \frac{\ell_R}{|R|!} \right).\]

Since the geometric mean of $\ell_T$ is related to the arithmetic mean of $\ln(\ell_T)$,
we introduce the sequence $w_T = \log(\ell_T / |T|!)$
and its generation function $W(z) = \sum_T w_T z^{|T|}$.
Using the latter recursive formula on $\ell_T$, we deduce:
\[W(z) = -L(z) + \sum_T \sum_{R \text{ child of the root of } T} w_R z^{|T|},\]
 where $L(z) =  \sum_T \log |T| z^{|T|} = \sum_{n\geq1} C_{n} \log (n)  z^n$ and
$C(z) = \sum_n C_n z^n$ is the generating function enumerating all trees.\\
By partitioning trees $T$ according to their number of root-children, we get
\[W(z) =  -L(z) + \sum_{r\geq 1} \sum_{R_1, \dots, R_r} \left( \sum_{i=1}^r w_{R_i}\right) z^{1+\sum_{j=1}^r |R_j|}.\]
Now, by symmetry of the trees $R_i$, we get:
\[W(z) = -L(z) + \sum_{r\geq 1} r \cdot \sum_{R_1, \dots, R_r} w_{R_1} z^{1+\sum_{j=1}^r |R_j|}
 = -L(z) + zU(z)\sum_{r\geq 1} r C^{r-1}(z).\]
We recognise $\sum_{r\geq 1} r C^{r-1}(z) = (1-C(z))^{-2}$, thus,
\[W(z) = \frac{-L(z)}{2}\left( 1+\frac{1}{\sqrt{1-4z}}\right).\]
In order to obtain the geometric mean width $\Gamma_n$,
we first extract the $n$-th coefficient of the previous product. Then we apply the exponential function on the result.
We then multiply it by $n!$ and take the $C_{n}$-th root of the result.
Finally $\bar{\Gamma}_n$ is equal to this result divided by $n$.\\

In order to approximate for $[z^n]W(z) = W_n$, we give an approximation $A(z)$ of $L(z)$:
\[A(z) = L(1/4) - \frac{\sqrt{1-4z}}{2}  \ln(\frac{1}{1-z}) +\frac{(\gamma +2\ln 2 - 2)\sqrt{1-4z}}{2},\]
where $\gamma$ is Euler's constant.
By using the two first terms in the development of Catalan numbers
(see Fact~\ref{fact:Catalan}) and formulas~\cite[p.~388]{FS09},
we obtain $L_n = A_n +\BigO (4^n\ n^{-5/2} \ \ln n)$.
Consequently,
\[W_n = -\frac{1}{2} \left( L(1/4)\frac{4^n}{\sqrt{\pi n}} - \frac{4^n}{2n} +
 \frac{4^{n-1} \ln n}{\sqrt{\pi n^3}} - L(1/4)\frac{4^{n-1}}{2\sqrt{\pi n^3}} +\BigO\left(\frac{4^n \ln n}{n^{5/2}}\right)\right).\]
Thus,
\[\frac{W_n}{C_{n}} = W_n \frac{\sqrt{\pi n^3}}{4^{n-1}}\left(1-\frac{3}{8n}+\BigO\left(\frac1{n^2}\right)\right)
 = - L(1/4)\cdot 2n + \sqrt{\pi n} - \frac{\ln n}{2} + L(1/4) +\BigO\left(\frac1{n^2}\right).\]
Finally we conclude:
\[\bar{\Gamma}_{n} = (n-1)! \cdot \exp\left(\frac{W_n}{C_{n}}\right)
\sim_{n\rightarrow\infty} \sqrt{2\pi} \frac{e^{\sqrt{\pi n} + L(1/4)}}{n} \left(\frac{n}{e^{1+2L(1/4)}}\right)^{n}.\]
\end{proof}

\subsection{The case of non-plane trees}

In classical concurrency theory, the pure merge operator often comes with commutativity laws, e.g.: $P \parallel Q \equiv Q \parallel P$.
From a combinatorial point of view, the idea is to consider the syntax and semantic trees as non-plane (or unordered) rooted trees.

Thankfully the non-plane analogous of the Catalan number is well known (cf.~\cite[p.~475--477]{FS09}):
\begin{fact}\label{fact:catalan_nonplane}
The specification of unlabelled non-plane rooted trees is $\T = \Z \times \MSet \T$.
The number $T_n$ of such trees of size $n$ is:
\[T_n \sim \frac{\gamma}{2\sqrt{\pi n^3}}\eta^{-n},\]
where $\eta\in [1/4, 1/e]$ and approximately $\eta\approx {\bf 0.3383218}$ and $\gamma \approx {\bf 1.559490}$.
\end{fact}

Compared to plane trees, no known closed form exists to characterize the symmetries
 involved in the non-plane case. One must indeed work with rather
complex approximations. Luckily, the increasing variant on non-plane trees have been studied in the model of 
 \emph{increasing Cayley trees}~\cite[p.~526--527]{FS09}:
\begin{fact}
The specification of increasing non-plane rooted trees is $\I = \Z^\text{\qed} \star \Set(\I)$.
The number $I_n$ of such trees of size $n$ is:
\[I_n = (n-1)!.\]
\end{fact}

\begin{theorem}
The mean number of concurrent runs built on non-plane syntax trees of size~$n$ is:
\[\bar{V}_n \sim_{n\rightarrow \infty} \frac{2\sqrt{2} \pi n}{\gamma}\left(\frac{n\eta}{e}\right)^{n},\]
where $\eta$ and $\gamma$ are introduced in Fact~\ref{fact:catalan_nonplane}.
\end{theorem}

Of course, we obtain different approximations for the plane vs. non-plane case.
The ratio $\bar{W}_n/\bar{V}_n$ is equivalent to $\gamma (2\eta)^{-n} / \sqrt{\pi n}$, which means
that although the exponential growths are not equivalent, the two asymptotic formulas follow a same universal \emph{shape}. This comparison between plane and non-plane combinatorial structures 
is a recurring theme in combinatorics. It has often been pointed out that in most cases the asymptotics look very similar. 
Citing Flajolet and Sedgewick  (cf.~\cite[p.~71--72]{FS09}): 
\begin{quote}
``(some) universal law governs the singularities of simple tree generating functions, either plane or non-plane''.
\end{quote}
Our study echoes quite faithfully such an intuition.

\section{Typical shape of process behaviours}\label{sec:shape}

Our goal in this section is to provide a more refined view of the
process behaviours by studying the typical shape of the semantic
trees. This study puts into light a new -- and, we think, interesting --
combinatorial class: the model of \emph{increasing admissible cuts}
(of plane trees). In the first part we recall the notion of
admissible cuts and define their increasing variant. This naturally leads to a generalization of the hook
length formula that enables the decomposition of a semantic
tree by levels. Based on this construction, we study experimentally the level
decomposition of semantic trees corresponding to syntax trees of a
size~$40$ (which yields semantic trees with more than
$10^{28}$ nodes !).  Finally, we discuss the mean number of nodes by
level, which is obtained by counting increasing admissible cuts. This
provides a fairly precise characterization of the typical shape for
process behaviours.

\subsection{Increasing admissible cuts}

\begin{figure}[thb]
\begin{center}
\scalebox{0.8} 
{
\hspace{4mm}
\begin{tikzpicture}[node distance=28pt]
\node (a) {$a$};
\node[below of=a] (b) {$b$};
\draw (a) -- (b);
\node[below left of=b] (c) {$c$};
\draw (b) -- (c);
\node[below right of=b] (d) {$d$};
\draw (b) -- (d);
\node[below left of=d] (e) {$e$};
\draw (d) -- (e);
\node[below right of=d] (f) {$f$};
\draw (d) -- (f);
\end{tikzpicture}
\hspace{4mm}
\begin{tikzpicture}[node distance=28pt]
[node distance=28pt]
\node (a) {$a$};
\node[below of=a] (b) {$b$};
\draw (a) -- (b);
\node[below left of=b] (c) {$c$};
\draw (b) -- (c);
\node[below right of=b] (d) {$d$};
\draw (b) -- (d);
\node[below left of=d] (e) {$e$};
\draw (d) -- (e);
\node[below right of=d] (f) {$\ $};
\end{tikzpicture}
\hspace{4mm}
\begin{tikzpicture}[node distance=28pt]
\node (a) {$a$};
\node[below of=a] (b) {$b$};
\draw (a) -- (b);
\node[below left of=b] (c) {$c$};
\draw (b) -- (c);
\node[below right of=b] (d) {$d$};
\draw (b) -- (d);
\node[below left of=d] (e) {$\ $};
\node[below right of=d] (f) {$f$};
\draw (d) -- (f);
\end{tikzpicture}
\hspace{4mm}
\begin{tikzpicture}[node distance=28pt]
\node (a) {$a$};
\node[below of=a] (b) {$b$};
\draw (a) -- (b);
\node[below left of=b] (c) {$\ $};
\node[below right of=b] (d) {$d$};
\draw (b) -- (d);
\node[below left of=d] (e) {$e$};
\draw (d) -- (e);
\node[below right of=d] (f) {$f$};
\draw (d) -- (f);
\end{tikzpicture}
\hspace{4mm}
\begin{tikzpicture}[node distance=28pt]
\node (a) {$a$};
\node[below of=a] (b) {$b$};
\draw (a) -- (b);
\node[below left of=b] (c) {$c$};
\draw (b) -- (c);
\node[below right of=b] (d) {$d$};
\draw (b) -- (d);
\node[below left of=d] (e) {$\ $};
\node[below right of=d] (f) {$\ $};
\end{tikzpicture}
}
\vspace*{3mm}

\scalebox{0.8}
{
\begin{tikzpicture}[node distance=28pt]
\node (a) {$a$};
\node[below of=a] (b) {$b$};
\draw (a) -- (b);
\node[below left of=b] (c) {$\ $};
\node[below right of=b] (d) {$d$};
\draw (b) -- (d);
\node[below left of=d] (e) {$e$};
\draw (d) -- (e);
\node[below right of=d] (f) {$\ $};
\end{tikzpicture}
\hspace{0mm}
\begin{tikzpicture}[node distance=28pt]
\node (a) {$a$};
\node[below of=a] (b) {$b$};
\draw (a) -- (b);
\node[below left of=b] (c) {$\ $};
\node[below right of=b] (d) {$d$};
\draw (b) -- (d);
\node[below left of=d] (e) {$\ $};
\node[below right of=d] (f) {$f$};
\draw (d) -- (f);
\end{tikzpicture}
\hspace{0mm}
\begin{tikzpicture}[node distance=28pt]
\node (a) {$a$};
\node[below of=a] (b) {$b$};
\draw (a) -- (b);
\node[below left of=b] (c) {$\ $};
\node[below right of=b] (d) {$d$};
\draw (b) -- (d);
\node[below left of=d] (e) {$\ $};
\node[below right of=d] (f) {$\ $};
\end{tikzpicture}
\hspace{0mm}
\begin{tikzpicture}[node distance=28pt]
\node (a) {$a$};
\node[below of=a] (b) {$b$};
\draw (a) -- (b);
\node[below left of=b] (c) {$c$};
\draw (b) -- (c);
\node[below right of=b] (d) {$\ $};
\node[below left of=d] (e) {$\ $};
\node[below right of=d] (f) {$\ $};
\end{tikzpicture}
\hspace{0mm}
\begin{tikzpicture}[node distance=28pt]
\node (a) {$a$};
\node[below of=a] (b) {$b$};
\draw (a) -- (b);
\node[below left of=b] (c) {$\ $};
\node[below right of=b] (d) {$\ $};
\node[below left of=d] (e) {$\ $};
\node[below right of=d] (f) {$\ $};
\end{tikzpicture}
\hspace{0mm}
\begin{tikzpicture}
[node distance=28pt]
\node (a) {$a$};
\node[below of=a] (b) {$\ $};
\node[below left of=b] (c) {$\ $};
\node[below right of=b] (d) {$\ $};
\node[below left of=d] (e) {$\ $};
\node[below right of=d] (f) {$\ $};
\end{tikzpicture}
}
\vspace*{-3mm}
\caption{The admissible cuts of the syntax tree of Figure~\ref{fig:task-and-shuffle}.}
\label{fig:substructures}
\end{center}
\end{figure}
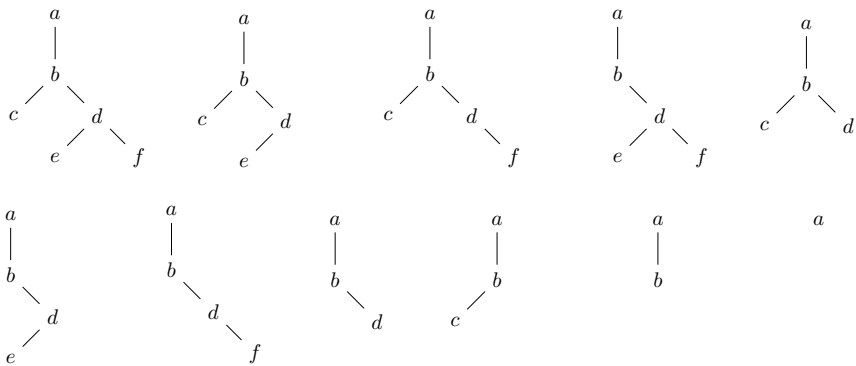

The notion of \emph{admissible cut} has been already studied in algebraic
 combinatorics, see for example~\cite{CK98}. The novelty here is the consideration
 of the increasingly labelled variant.

\begin{definition}
Let $T$ be a tree of size $n$. An \textbf{admissible cut} of $T$
of size $k=n-i$  $(0\leq i < n)$ is a tree obtained by starting with $T$ and removing recursively~$i$ leaves from it.
An \textbf{increasing admissible cut}  of $T$ of size $k$ is an admissible cut of size~$k$ of $T$ that is increasingly
labelled.
\end{definition}

Figure~\ref{fig:substructures} depicts the set of all admissible cuts
for the syntax tree $T$ of Figure~\ref{fig:task-and-shuffle}. 
We remark that the tree $T$ is itself an admissible cut of $T$.

To establish a link with increasing admissible cuts, we first make a simple
albeit important observation.

\begin{proposition}
Let $T$ be a syntax tree of size $n$. Any run prefix of length $k$ $(1 \leq k \leq n)$ in $\Shuf(T)$ is uniquely encoded
by an admissible cut of $T$ of size $k$.
\end{proposition}

\begin{proof}
We proceed by finite induction on $k$.
For $k=1$ there is a single run prefix of length $k=1$ with the root of $T$ and the corresponding
admissible cut is the root node, which only has one increasing labelling. Now suppose that
the property holds for run prefixes of length $k,~1<k< n$, let us show that it also holds for run prefixes of  length~$k+1$.
By hypothesis of induction, any run prefix $\sigma_k$ of length $k$ is encoded by a given admissible cut of size $k$.
 Let us denote by $S(\sigma_k)$ this admissible cut. Now, 
any prefix $\sigma_{k+1}$ of length $k+1$ is obtained by appending an action $\alpha$ to a prefix $\sigma_k$ of length $k$.
For $\sigma_{k+1}$ to be a valid prefix, $\alpha$ must corresponds to a node in $T$ that is a direct child of one of the nodes of $S(\sigma_k)$. Thus we obtain
a unique $S(\sigma_{k+1})$ as $S(\sigma_k)$ completed by a single leaf $\alpha$.
\end{proof}

For example the run prefixes $\langle a,b,c,d \rangle$ and $\langle a,b,d,e \rangle$ are encoded by both first admissible cuts of size~$4$
depicted on Figure~\ref{fig:substructures}. 

This result leads to a fundamental connection with increasing admissible cuts.

\begin{proposition}
Let $T$ be a syntax tree of size $n$. The number of run prefixes of length $k$ $(1 \leq k \leq n)$ in $\Shuf(T)$
is the number of increasing labellings of the admissible cuts of $T$ of size $k$.
\end{proposition}

\begin{proof}
This is obtained by a trivial order-theoretic argument. Each admissible cut is a tree-poset and thus the number
of run it encodes is the number of its linear extensions.
\end{proof}

For example, there are three admissible cuts of size $4$ in Figure~\ref{fig:substructures}. The first one admits two increasing labellings and the other ones have a single labelling. This gives $2+1+1=4$ run prefixes of length $4$ for the syntax tree $T$ of Figure~\ref{fig:task-and-shuffle}. 
Now, we observe that this is also the number of nodes at level $3$ in the corresponding semantic tree. And this of course generalizes: the number of run prefixes of length $k$  corresponds to the number of nodes at level $k-1$ in the semantic tree.  

From this we can characterize precisely the number of nodes by level thanks to a generalization of the hook-length formula.

\begin{corollary}\label{sizeshuff}
Let $T$ be a process tree of size $n$. The number of nodes at level $n-1-i$  $(0 \leq i < n)$ of $\Shuf(T)$ is:
\[n^i_T = \sum_{\substack{S \text{ admissible cut of } T \\ |S|=n-i}} \ell_S,\]
\indent \hfill where $\ell_S$ is the hook-length formula applied to the admissible cut $S$ (cf. Fact~\ref{fact:hook}).

Moreover, the total number of nodes of $\Shuf(T)$ is:
\[n_T = \sum^{n-1}_{i=0} n^i_T = \sum_{S \text{ admissible cut of } T} \ell_S.\]
\end{corollary}

\subsection{Level decomposition}

Before working an exact formula for the mean number of nodes by level,
 we can take advantage of Corollary~\ref{sizeshuff} to compute the shape
 of some typical semantic trees.

\subsubsection{Experimental study}

Our experiments consists in generating
 uniformly at random some syntax trees (using our \emph{arbogen tool}\footnote{\url{https://github.com/fredokun/arbogen}})
 of size $n$ for $n$ not too small. Then we can compute $n_T$ as defined above by first listing all the
admissible cuts of $T$.

However we cannot take syntax tree with a size $n$ very large, given the following result.

\begin{observation}\label{prop:mean_nb_substructures}
The mean number $\bar{m}_n$ of admissible cuts of trees of size $n$ satisfies:
\[\bar{m}_n 
\sim_{n\rightarrow\infty} \frac{1}{\sqrt{15}}\left(\frac{5}{2}\right)^{2n}.\]
\end{observation}
\begin{proof}
Let us denote by $M(z)$ the ordinary generating function enumerating the multiset of admissible cuts
of all trees. More precisely, we get $M(z)=\sum_{n\in\mathbb{N}}M_nz^n$ where
 $M_n=\sum_{T;|T|=n}\sum_{S \text{ adm. cut of } T} 1$.
The tag $\Z$ marks the nodes of the tree carrying the admissible cut.
The generating function $C(z)$ enumerates all trees.
The specification of $\M$ is $\Z \times Seq( \M \cup \C)$.
In fact, an admissible cut is a root and a sequence of children that are either admissible cuts,
or trees that corresponds to a branch of the original tree that has entirely disappeared.
Consequently, $M(z)$ satisfies the following equation that can be easily solve:
\[M(z) = \frac{z}{1-M(z)-C(z)}, \hspace{1cm} 
M(z) = \frac{1 + \sqrt{1-4z} - \sqrt{2 - 20z + 2\sqrt{1-4z}}}{4}.\]
The singularities of $M(z)$ are  $1/4$ and $4/25$: the latter one is the dominant.
The generating function is analytic in a $\Delta$-domain around $4/25$ because
of the square-root type of the dominant singularity.
By using transfer lemmas~\cite[p.~392]{FS09}, we get the asymptotic behaviour.
\end{proof}

Using the same kind of method than one of the following (Section~\ref{sec:size}),
we can exhibit the P-recurrence satisfied by the cumulative number of 
admissible cuts $m_n$:
\begin{multline*}
(-500n+2000n^3)m_n+(120-220n-1380n^2-920n^3)m_{n+1}-(1488+1626n+387n^2-21n^3)m_{n+2}\\
 + (1104+1088n+351n^2+37n^3)m_{n+3}-(168+146n+42n^2+4n^3)m_{n+4} =0,
\end{multline*}
with $m_0= 0, m_1= 1, m_2 = 2$ and $m_3= 7$.
This sequence is registered by OEIS at A007852.

As a consequence we must be particularly careful when computing the shape of a semantic tree
 in practice using our generalization of the hook-length formula for increasing admissible cuts.
However, for syntax trees of a size $n\leq 40$ we are able to compute the level
decomposition within a couple of days using a fast computer\footnote{The computer used for the experiment is a bi-Xeon 5420 machine
with 8  cores running at 2.5Ghz each, equipped with 20GB of RAM and running linux.}. This must be
compared to the mean width of these trees: $\bar{W}^0_{40} > 1.48\cdot10^{36}$ !

\begin{figure}[htb]
\begin{center}
\scalebox{0.7}{
\begin{tabular}{c c | c}
\includegraphics[height=80mm, width=50mm]{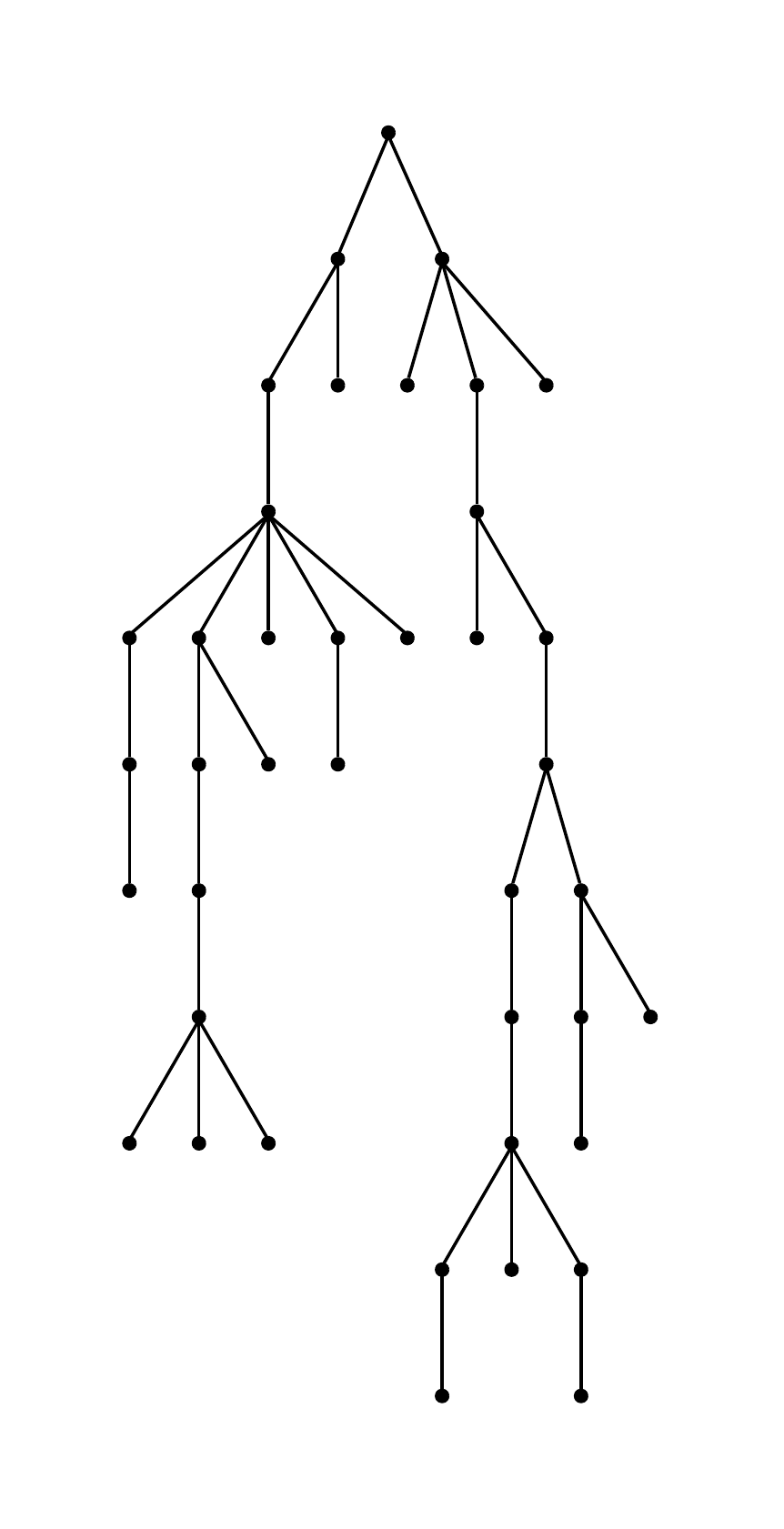} 
&
\includegraphics[height=80mm, width=80mm]{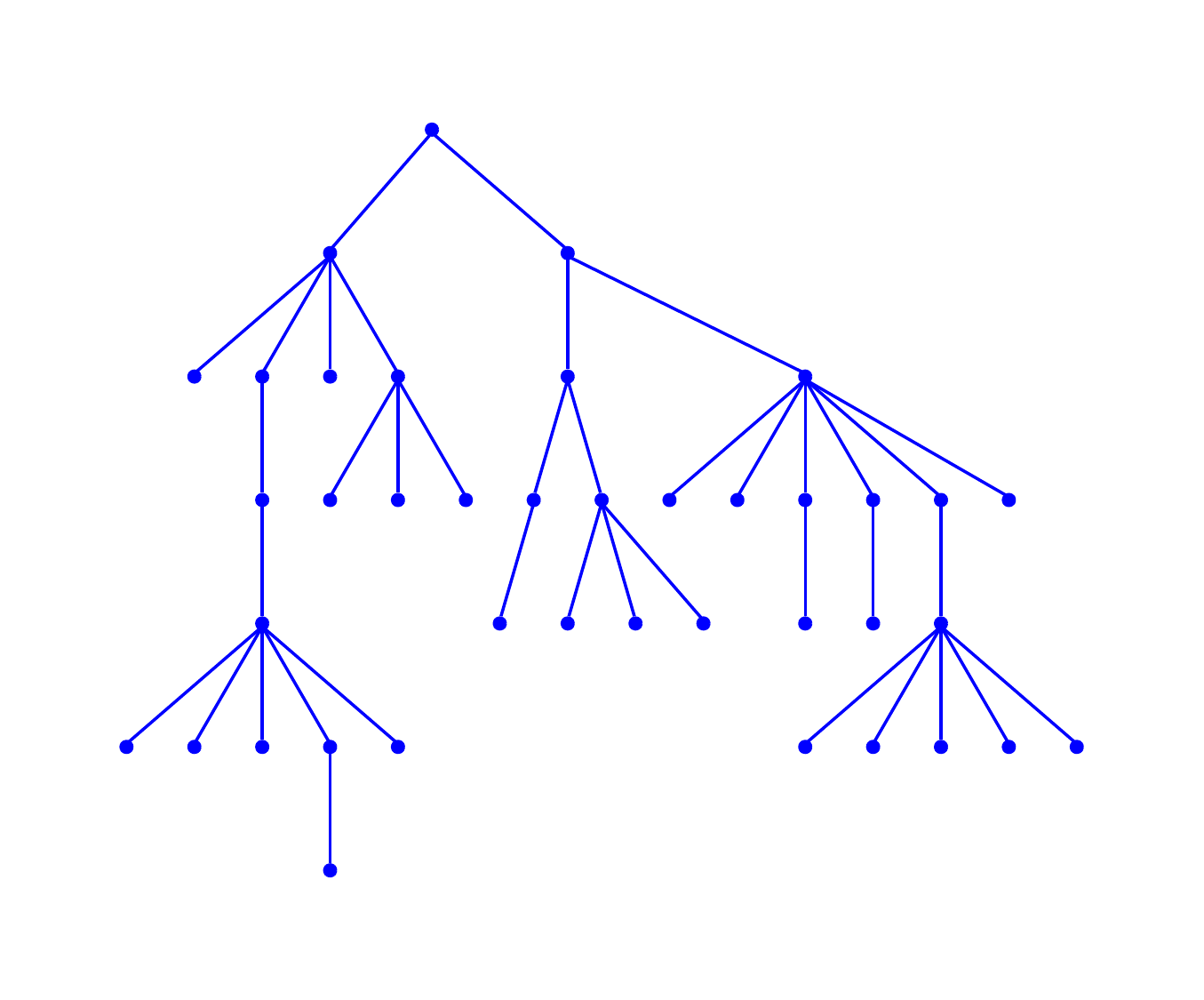}
&
\begin{tikzpicture}[scale=0.2]
\draw[ultra thin, gray!60] ( 0 , 0 ) -- ( 0 , 39 );
\draw[ultra thin, gray!60] ( 3.80000000000000 , 0 ) -- ( 3.80000000000000 , 39 );
\draw[ultra thin, gray!60] ( 6.33333333333333 , 0 ) -- ( 6.33333333333333 , 39 );
\draw[ultra thin, gray!60] ( 8.02222222222222 , 0 ) -- ( 8.02222222222222 , 39 );
\draw[ultra thin, gray!60] ( 9.14814814814815 , 0 ) -- ( 9.14814814814815 , 39 );
\draw[ultra thin, gray!60] ( 9.89876543209877 , 0 ) -- ( 9.89876543209877 , 39 );
\draw[ultra thin, gray!60] ( 10.3991769547325 , 0 ) -- ( 10.3991769547325 , 39 );
\draw[ultra thin, gray!60] ( 10.7327846364883 , 0 ) -- ( 10.7327846364883 , 39 );
\draw[ultra thin, gray!60] ( 10.9551897576589 , 0 ) -- ( 10.9551897576589 , 39 );
\draw[ultra thin, gray!60] ( 11.1034598384393 , 0 ) -- ( 11.1034598384393 , 39 );
\draw[ultra thin, gray!60] ( 11.2023065589595 , 0 ) -- ( 11.2023065589595 , 39 );
\draw[ultra thin, gray!60] ( 11.2682043726397 , 0 ) -- ( 11.2682043726397 , 39 );
\draw[ultra thin, gray!60] ( 11.3121362484264 , 0 ) -- ( 11.3121362484264 , 39 );
\draw[ultra thin, gray!60] ( 11.3414241656176 , 0 ) -- ( 11.3414241656176 , 39 );
\draw[ultra thin, gray!60] ( 11.3609494437451 , 0 ) -- ( 11.3609494437451 , 39 );
\draw[ultra thin, gray!60] ( 11.3739662958301 , 0 ) -- ( 11.3739662958301 , 39 );
\draw[ultra thin, gray!60] ( 11.3826441972200 , 0 ) -- ( 11.3826441972200 , 39 );
\draw[ultra thin, gray!60] ( 11.3884294648134 , 0 ) -- ( 11.3884294648134 , 39 );
\draw[ultra thin, gray!60] ( 11.3922863098756 , 0 ) -- ( 11.3922863098756 , 39 );
\draw[ultra thin, gray!60] ( 11.3948575399170 , 0 ) -- ( 11.3948575399170 , 39 );
\draw[ultra thin, gray!60] ( 11.3965716932780 , 0 ) -- ( 11.3965716932780 , 39 );
\draw[ultra thin, gray!60] ( 11.3977144621854 , 0 ) -- ( 11.3977144621854 , 39 );
\draw[ultra thin, gray!60] ( 11.3984763081236 , 0 ) -- ( 11.3984763081236 , 39 );
\draw[ultra thin, gray!60] ( 11.3989842054157 , 0 ) -- ( 11.3989842054157 , 39 );
\draw[ultra thin, gray!60] ( 11.3993228036105 , 0 ) -- ( 11.3993228036105 , 39 );
\draw[ultra thin, gray!60] ( 11.3995485357403 , 0 ) -- ( 11.3995485357403 , 39 );
\draw[ultra thin, gray!60] ( 11.3996990238269 , 0 ) -- ( 11.3996990238269 , 39 );
\draw[ultra thin, gray!60] ( 11.3997993492179 , 0 ) -- ( 11.3997993492179 , 39 );
\draw[ultra thin, gray!60] ( 11.3998662328119 , 0 ) -- ( 11.3998662328119 , 39 );
\draw[ultra thin, gray!60] ( 11.3999108218746 , 0 ) -- ( 11.3999108218746 , 39 );
\draw[ultra thin, gray!60] ( 11.3999405479164 , 0 ) -- ( 11.3999405479164 , 39 );
\draw[ultra thin, gray!60] ( 11.3999603652776 , 0 ) -- ( 11.3999603652776 , 39 );
\draw[ultra thin, gray!60] ( 11.3999735768517 , 0 ) -- ( 11.3999735768517 , 39 );
\draw[ultra thin, gray!60] ( 11.3999823845678 , 0 ) -- ( 11.3999823845678 , 39 );
\draw[ultra thin, gray!60] ( 11.3999882563785 , 0 ) -- ( 11.3999882563785 , 39 );
\draw[ultra thin, gray!60] ( 11.3999921709190 , 0 ) -- ( 11.3999921709190 , 39 );
\draw[ultra thin, gray!60] ( 11.3999947806127 , 0 ) -- ( 11.3999947806127 , 39 );
\draw[ultra thin, gray!60] ( 11.3999965204085 , 0 ) -- ( 11.3999965204085 , 39 );
\draw[ultra thin, gray!60] ( 11.3999976802723 , 0 ) -- ( 11.3999976802723 , 39 );
\draw[ultra thin, gray!60] ( 11.3999984535149 , 0 ) -- ( 11.3999984535149 , 39 );
\draw[ultra thin, gray!60] ( 11.3999989690099 , 0 ) -- ( 11.3999989690099 , 39 );
\draw[ultra thin, gray!60] ( 11.3999993126733 , 0 ) -- ( 11.3999993126733 , 39 );
\draw[ultra thin, gray!60] ( 11.3999995417822 , 0 ) -- ( 11.3999995417822 , 39 );
\draw[ultra thin, gray!60] ( 11.3999996945215 , 0 ) -- ( 11.3999996945215 , 39 );
\draw[ultra thin, gray!60] ( 11.3999997963476 , 0 ) -- ( 11.3999997963476 , 39 );
\draw[ultra thin, gray!60] ( 11.3999998642318 , 0 ) -- ( 11.3999998642318 , 39 );
\draw[ultra thin, gray!60] ( 11.3999999094878 , 0 ) -- ( 11.3999999094878 , 39 );
\draw[ultra thin, gray!60] ( 11.3999999396586 , 0 ) -- ( 11.3999999396586 , 39 );
\draw[ultra thin, gray!60] ( 11.3999999597724 , 0 ) -- ( 11.3999999597724 , 39 );
\draw[ultra thin, gray!60] ( 11.3999999731816 , 0 ) -- ( 11.3999999731816 , 39 );
\draw[ultra thin, gray!60] ( 11.3999999821211 , 0 ) -- ( 11.3999999821211 , 39 );
\draw[ultra thin, gray!60] ( 11.3999999880807 , 0 ) -- ( 11.3999999880807 , 39 );
\draw[ultra thin, gray!60] ( 11.3999999920538 , 0 ) -- ( 11.3999999920538 , 39 );
\draw[ultra thin, gray!60] ( 11.3999999947025 , 0 ) -- ( 11.3999999947025 , 39 );
\draw[ultra thin, gray!60] ( 11.3999999964684 , 0 ) -- ( 11.3999999964684 , 39 );
\draw[ultra thin, gray!60] ( 11.3999999976456 , 0 ) -- ( 11.3999999976456 , 39 );
\draw[ultra thin, gray!60] ( 11.3999999984304 , 0 ) -- ( 11.3999999984304 , 39 );
\draw[ultra thin, gray!60] ( 11.3999999989536 , 0 ) -- ( 11.3999999989536 , 39 );
\draw[ultra thin, gray!60] ( 11.3999999993024 , 0 ) -- ( 11.3999999993024 , 39 );
\draw[ultra thin, gray!60] ( -4 , 0 ) -- ( -4 , 39 );
\draw[ultra thin, gray!60] ( -6.60586319218241 , 0 ) -- ( -6.60586319218241 , 39 );
\draw[ultra thin, gray!60] ( -8.30349393627519 , 0 ) -- ( -8.30349393627519 , 39 );
\draw[ultra thin, gray!60] ( -9.40944230376234 , 0 ) -- ( -9.40944230376234 , 39 );
\draw[ultra thin, gray!60] ( -10.1299298395846 , 0 ) -- ( -10.1299298395846 , 39 );
\draw[ultra thin, gray!60] ( -10.5993028270909 , 0 ) -- ( -10.5993028270909 , 39 );
\draw[ultra thin, gray!60] ( -10.9050832749778 , 0 ) -- ( -10.9050832749778 , 39 );
\draw[ultra thin, gray!60] ( -11.1042887784872 , 0 ) -- ( -11.1042887784872 , 39 );
\draw[ultra thin, gray!60] ( -11.2340643508060 , 0 ) -- ( -11.2340643508060 , 39 );
\draw[ultra thin, gray!60] ( -11.3186086975935 , 0 ) -- ( -11.3186086975935 , 39 );
\draw[ultra thin, gray!60] ( -11.3736864479436 , 0 ) -- ( -11.3736864479436 , 39 );
\draw[ultra thin, gray!60] ( -11.4095677185301 , 0 ) -- ( -11.4095677185301 , 39 );
\draw[ultra thin, gray!60] ( -11.4329431391075 , 0 ) -- ( -11.4329431391075 , 39 );
\draw[ultra thin, gray!60] ( -11.4481714261287 , 0 ) -- ( -11.4481714261287 , 39 );
\draw[ultra thin, gray!60] ( -11.4580921342858 , 0 ) -- ( -11.4580921342858 , 39 );
\draw[ultra thin, gray!60] ( -11.4645551363425 , 0 ) -- ( -11.4645551363425 , 39 );
\draw[ultra thin, gray!60] ( -11.4687655611352 , 0 ) -- ( -11.4687655611352 , 39 );
\draw[ultra thin, gray!60] ( -11.4715085088829 , 0 ) -- ( -11.4715085088829 , 39 );
\draw[ultra thin, gray!60] ( -11.4732954455263 , 0 ) -- ( -11.4732954455263 , 39 );
\draw[ultra thin, gray!60] ( -11.4744595736328 , 0 ) -- ( -11.4744595736328 , 39 );
\draw[ultra thin, gray!60] ( -11.4752179632787 , 0 ) -- ( -11.4752179632787 , 39 );
\draw[ultra thin, gray!60] ( -11.4757120281946 , 0 ) -- ( -11.4757120281946 , 39 );
\draw[ultra thin, gray!60] ( -11.4760338945893 , 0 ) -- ( -11.4760338945893 , 39 );
\draw[ultra thin, gray!60] ( -11.4762435795370 , 0 ) -- ( -11.4762435795370 , 39 );
\draw[ultra thin, gray!60] ( -11.4763801821088 , 0 ) -- ( -11.4763801821088 , 39 );
\draw[ultra thin, gray!60] ( -11.4764691740122 , 0 ) -- ( -11.4764691740122 , 39 );
\draw[ultra thin, gray!60] ( -11.4765271491936 , 0 ) -- ( -11.4765271491936 , 39 );
\draw[ultra thin, gray!60] ( -11.4765649180415 , 0 ) -- ( -11.4765649180415 , 39 );
\draw[ultra thin, gray!60] ( -11.4765895231541 , 0 ) -- ( -11.4765895231541 , 39 );

\draw (0,-0.2) node[below]{$0$};
\draw (4,0) node[below]{$10^{13}$};
\draw (8,0) node[below]{$10^{23}$};
\draw (11.5,0) node[below]{$10^{36}$};

\draw[ultra thin, gray!60] ( -11.5,  0 ) -- ( 11.5,  0 );
\draw[ultra thin, gray!60] ( -11.5,  4 ) -- ( 11.5,  4 );
\draw[ultra thin, gray!60] ( -11.5,  8 ) -- ( 11.5,  8 );
\draw[ultra thin, gray!60] ( -11.5,  12 ) -- ( 11.5,  12 );
\draw[ultra thin, gray!60] ( -11.5,  16 ) -- ( 11.5,  16 );
\draw[ultra thin, gray!60] ( -11.5,  20 ) -- ( 11.5,  20 );
\draw[ultra thin, gray!60] ( -11.5,  24 ) -- ( 11.5,  24 );
\draw[ultra thin, gray!60] ( -11.5,  28 ) -- ( 11.5,  28 );
\draw[ultra thin, gray!60] ( -11.5,  32 ) -- ( 11.5,  32 );
\draw[ultra thin, gray!60] ( -11.5,  36 ) -- ( 11.5,  36 );

\draw (-12.2,39) node[left]{$0$};
\draw (-12.2,32) node[left]{$7$};
\draw (-12.2,24) node[left]{$15$};
\draw (-12.2,16) node[left]{$23$};
\draw (-12.2,8) node[left]{$31$};
\draw (-12.2,0) node[left]{$39$};

\draw[thick](-9.15877553859415, 0) -- (-9.15877553859415, 1) -- (-9.06593795509456, 2) --
(-8.92014361971971, 3) -- (-8.73793681399964, 4) -- (-8.52866051766464, 5) --
(-8.29866472278052, 6) -- (-8.05294455561430, 7) -- (-7.79580736788745, 8) --
(-7.53093621579451, 9) -- (-7.26111730125745, 10) -- (-6.98804497069023,11) -- (-6.71253930458880, 12) -- (-6.43511469923669, 13) --
(-6.15651597474352, 14) -- (-5.87783237238222, 15) -- (-5.60012000450140,
16) -- (-5.32395265302678, 17) -- (-5.04941915161766, 18) --
(-4.77649675252845, 19) -- (-4.50531358103039, 20) -- (-4.23607029169226,
21) -- (-3.96878757503394, 22) -- (-3.70311245643660, 23) --
(-3.43831888232790, 24) -- (-3.17347340091245, 25) -- (-2.90764534086813,
26) -- (-2.64008833795199, 27) -- (-2.37039900239825, 28) --
(-2.09885560581268, 29) -- (-1.82710324750891, 30) -- (-1.55833090236987,
31) -- (-1.29831948822203, 32) -- (-1.05750354575514, 33) --
(-0.841111575487845, 34) -- (-0.639729768059481, 35) -- (-0.449122200851964,
36) -- (-0.271508283125423, 37) -- (-0.0967132018086354, 38) --
(-0.000000000000000, 39) -- (0.000000000000000, 39) -- (0.0967132018086354,
38) -- (0.271508283125423, 37) -- (0.449122200851964, 36) --
(0.639729768059481, 35) -- (0.841111575487845, 34) -- (1.05750354575514,
33) -- (1.29831948822203, 32) -- (1.55833090236987, 31) -- (1.82710324750891,
30) -- (2.09885560581268, 29) -- (2.37039900239825, 28) -- (2.64008833795199,
27) -- (2.90764534086813, 26) -- (3.17347340091245, 25) -- (3.43831888232790,
24) -- (3.70311245643660, 23) -- (3.96878757503394, 22) -- (4.23607029169226,
21) -- (4.50531358103039, 20) -- (4.77649675252845, 19) -- (5.04941915161766,
18) -- (5.32395265302678, 17) -- (5.60012000450140, 16) -- (5.87783237238222,
15) -- (6.15651597474352, 14) -- (6.43511469923669, 13) -- (6.71253930458880,
12) -- (6.98804497069023, 11) -- (7.26111730125745, 10) -- (7.53093621579451,
9) -- (7.79580736788745, 8) -- (8.05294455561430, 7) -- (8.29866472278052, 6) --
(8.52866051766464, 5) -- (8.73793681399964, 4) -- (8.92014361971971, 3) --
(9.06593795509456, 2) -- (9.15877553859415, 1) -- (9.15877553859415, 0);

\draw[thick, blue](-11.2409471004766, 0) -- (-11.2409471004766, 1) -- (-11.1455091585125, 2) --
(-10.9947384368214, 3) -- (-10.8050499653933, 4) -- (-10.5854426555900, 5) --
(-10.3416214205066, 6) -- (-10.0775522188107, 7) -- (-9.79619623698555, 8) --
(-9.49991907602552, 9) -- (-9.19075233751775, 10) -- (-8.87057685470337,
11) -- (-8.54124787111785, 12) -- (-8.20465285998964, 13) --
(-7.86267798291040, 14) -- (-7.51707375246825, 15) -- (-7.16926518258794,
16) -- (-6.82021427676951, 17) -- (-6.47043655193480, 18) --
(-6.12017055526051, 19) -- (-5.76959018104596, 20) -- (-5.41893936755457,
21) -- (-5.06854959007363, 22) -- (-4.71878773341555, 23) --
(-4.37001258369071, 24) -- (-4.02258332448429, 25) -- (-3.67689940643003,
26) -- (-3.33341501502441, 27) -- (-2.99260091031986, 28) --
(-2.65491840422400, 29) -- (-2.32097571596291, 30) -- (-1.99202086230698,
31) -- (-1.67057502024022, 32) -- (-1.36039279113604, 33) --
(-1.06513369829862, 34) -- (-0.787202072576296, 35) -- (-0.527999111893714,
36) -- (-0.290139605425906, 37) -- (-0.0967132018086354, 38) --
(-0.000000000000000, 39) -- (0.000000000000000, 39) -- (0.0967132018086354,
38) -- (0.290139605425906, 37) -- (0.527999111893714, 36) --
(0.787202072576296, 35) -- (1.06513369829862, 34) -- (1.36039279113604, 33) --
(1.67057502024022, 32) -- (1.99202086230698, 31) -- (2.32097571596291, 30) --
(2.65491840422400, 29) -- (2.99260091031986, 28) -- (3.33341501502441, 27) --
(3.67689940643003, 26) -- (4.02258332448429, 25) -- (4.37001258369071, 24) --
(4.71878773341555, 23) -- (5.06854959007363, 22) -- (5.41893936755457, 21) --
(5.76959018104596, 20) -- (6.12017055526051, 19) -- (6.47043655193480, 18) --
(6.82021427676951, 17) -- (7.16926518258794, 16) -- (7.51707375246825, 15) --
(7.86267798291040, 14) -- (8.20465285998964, 13) -- (8.54124787111785, 12) --
(8.87057685470337, 11) -- (9.19075233751775, 10) -- (9.49991907602552, 9) --
(9.79619623698555, 8) -- (10.0775522188107, 7) -- (10.3416214205066, 6) --
(10.5854426555900, 5) -- (10.8050499653933, 4) -- (10.9947384368214, 3) --
(11.1455091585125, 2) -- (11.2409471004766, 1) -- (11.2409471004766, 0) ;

\draw[thick, red, dashed](-11.6209654676476, 0) -- (-11.6209654676476, 1) -- (-11.5260526476784, 2) --
(-11.3764378418456, 3) -- (-11.1886307609987, 4) -- (-10.9717169253753, 5) --
(-10.7314780833053, 6) -- (-10.4719366438370, 7) -- (-10.1960677665525, 8) --
(-9.90617402490231, 9) -- (-9.60410171102600, 10) -- (-9.29137471688983,
11) -- (-8.96928206828212, 12) -- (-8.63893779440435, 13) --
(-8.30132349088075, 14) -- (-7.95731964908815, 15) -- (-7.60772948645458,
16) -- (-7.25329767473529, 17) -- (-6.89472556910989, 18) --
(-6.53268405694691, 19) -- (-6.16782484691539, 20) -- (-5.80079083856706,
21) -- (-5.43222611199216, 22) -- (-5.06278603715824, 23) --
(-4.69314801562449, 24) -- (-4.32402343590891, 25) -- (-3.95617156087744,
26) -- (-3.59041629859524, 27) -- (-3.22766718785912, 28) --
(-2.86894654798733, 29) -- (-2.51542576890516, 30) -- (-2.16847547671384,
31) -- (-1.82973745444543, 32) -- (-1.50123212193657, 33) --
(-1.18552729039416, 34) -- (-0.886019895375025, 35) -- (-0.607445410886010,
36) -- (-0.356905651386429, 37) -- (-0.146308961401111, 38) --
(-0.000000000000000, 39) -- (0.000000000000000, 39) -- (0.146308961401111,
38) -- (0.356905651386429, 37) -- (0.607445410886010, 36) --
(0.886019895375025, 35) -- (1.18552729039416, 34) -- (1.50123212193657, 33) --
(1.82973745444543, 32) -- (2.16847547671384, 31) -- (2.51542576890516, 30) --
(2.86894654798733, 29) -- (3.22766718785912, 28) -- (3.59041629859524, 27) --
(3.95617156087744, 26) -- (4.32402343590891, 25) -- (4.69314801562449, 24) --
(5.06278603715824, 23) -- (5.43222611199216, 22) -- (5.80079083856706, 21) --
(6.16782484691539, 20) -- (6.53268405694691, 19) -- (6.89472556910989, 18) --
(7.25329767473529, 17) -- (7.60772948645458, 16) -- (7.95731964908815, 15) --
(8.30132349088075, 14) -- (8.63893779440435, 13) -- (8.96928206828212, 12) --
(9.29137471688983, 11) -- (9.60410171102600, 10) -- (9.90617402490231, 9) --
(10.1960677665525, 8) -- (10.4719366438370, 7) -- (10.7314780833053, 6) --
(10.9717169253753, 5) -- (11.1886307609987, 4) -- (11.3764378418456, 3) --
(11.5260526476784, 2) -- (11.6209654676476, 1) -- (11.6209654676476, 0);

\end{tikzpicture}
\end{tabular}
}
\caption{Typical trees of size~$40$ and their semantic-tree profile behind the average profile.}
\label{fig:shape40}
\end{center}
\end{figure}
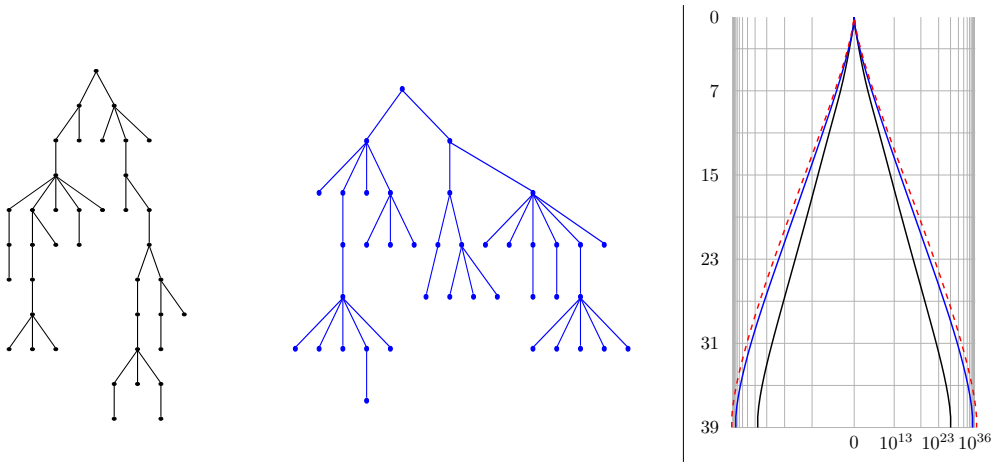

For the two syntax trees depicted on the left of Figure~\ref{fig:shape40}, the shape of the corresponding semantic tree is depicted on the right.
We use a logarithmic scale of the horizontal axis so that the exponential fringes become lines.
The semantic trees are of size larger than $8.74\cdot10^{28}$ for the one corresponding to the left process
tree and $2.66\cdot10^{35}$ for the second one. These correspond to the two plain lines in the figure.
The dashed lines correspond to the theoretical computation of the mean as explained in the next section. We
can see that it is almost reached by the shape of the second tree. We also remark that both the trees have a semantic size that is smaller than the average. To analyse this particularity, we have sampled more than fifty typical process trees of size~$40$ (of course with uniform probability among all trees
of size~$40$). The results are fairly interesting. All the shapes that we computed follow the same kind
of curve as the average. However, almost all process trees have a semantic-tree size that is much smaller than the
average, ($\approx 4.06\cdot10^{36}$). Indeed, most of them have a size that belong to $[10^{28}, 10^{35}]$,
and a single one has a size larger than the average (it is approximately twice as large).

This observation let us conjecture that a only a very few special syntax trees accounts for the largest increase
of the semantic size.
These are probably process trees whose nodes have a large arity.
The simplest one is the process tree with a single internal node. In the case of size~$40$, its semantic correspondence
has size larger than $2.03\cdot10^{46}$.
In fact, the combinatorial explosion in the worst case increases like a factorial function. Since the Catalan numbers
(that count syntax trees) do not increase that quickly, the ``worst'' syntax trees (the one whose semantic-tree size
is largest) do really influence the average measures.

\subsubsection{Mean number of nodes by level}

We may now describe one of the fundamental results of this paper: a close formula for the mean number
 of nodes at each level of a semantic tree.

\begin{theorem}\label{prop:bord}
The mean number of nodes at level $n-i-1$, for $i\in\{0,\dots,n-1\}$ in a semantic tree
corresponding to a syntax tree of size $n$ is:
\[\bar{W}^i_n = \frac{2^i \ (2n-2i-1)!\ (n-1)!}{(2n-i-1)!\ (n-i-1)!} \cdot \frac{n!}{2^{n-1}\ i!}.\]
\end{theorem}

\begin{proof}
Let $n,i$ be two integers such that $0 \leq i< n$.
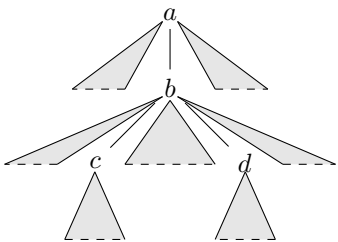
\begin{figure}[htb]
\begin{center}

\begin{tikzpicture}[node distance=28pt]
\node (a) {$a$};
\node[below of=a] (b) {$b$};
\draw (a) -- (b);
\node[below left of=b, node distance=40pt] (c) {$c$};
\draw (b) -- (c);
\node[below right of=b, node distance=40pt] (d) {$d$};
\draw (b) -- (d);
\fill[gray!20] (-0.1,-0.1) -- (-0.6,-1) -- (-1.3,-1) -- cycle;
\fill[gray!20] (0.1,-0.1) -- (0.6,-1) -- (1.3,-1) -- cycle;
\fill[gray!20] (-0.1,-1.1) -- (-1.5,-2) -- (-2.2,-2) -- cycle;
\fill[gray!20] (0,-1.15) -- (-0.6,-2) -- (0.6,-2) -- cycle;
\fill[gray!20] (0.1,-1.1) -- (1.5,-2) -- (2.2,-2) -- cycle;
\fill[gray!20] (-1,-2.1) -- (-1.4,-3) -- (-0.6,-3) -- cycle;
\fill[gray!20] (1,-2.1) -- (1.4,-3) -- (0.6,-3) -- cycle;
\draw (-0.1,-0.1) -- (-0.6,-1);
\draw (-0.1,-0.1) -- (-1.3,-1);
\draw[dashed] (-0.6,-1) -- (-1.3,-1);
\draw (0.1,-0.1) -- (0.6,-1);
\draw (0.1,-0.1) -- (1.3,-1);
\draw[dashed] (0.6,-1) -- (1.3,-1);
\draw (-0.1,-1.1) -- (-1.5,-2);
\draw (-0.1,-1.1) -- (-2.2,-2);
\draw[dashed] (-1.5,-2) -- (-2.2,-2);
\draw (0,-1.15) -- (-0.6,-2);
\draw (0,-1.15) -- (0.6,-2);
\draw[dashed] (-0.6,-2) -- (0.6,-2);
\draw (0.1,-1.1) -- (1.5,-2);
\draw (0.1,-1.1) -- (2.2,-2);
\draw[dashed] (1.5,-2) -- (2.2,-2);
\draw (-1,-2.1) -- (-1.4,-3);
\draw (-1,-2.1) -- (-0.6,-3);
\draw[dashed] (-1.4,-3) -- (-0.6,-3);
\draw (1,-2.1) -- (1.4,-3);
\draw (1,-2.1) -- (0.6,-3);
\draw[dashed] (1.4,-3) -- (0.6,-3);
\end{tikzpicture}
\caption{An admissible cut and the places where it can be enriched.}
\label{fig:enriched-cut}
\end{center}
\end{figure}
A direct corollary of Corollary~\ref{sizeshuff} gives
the cumulated number $W^i_n$  of nodes at level $n-i-1$ in semantic trees issued of syntax trees of size $n$
to be equal to the sum on all increasingly labelled admissible cuts of size $n-i$ from syntax trees of size $n$.
As in the previous section, let us denote by $\G$ the combinatorial class of increasing trees and thus $G_{n-i}$
the number of increasing trees of size~$n-i$.
An admissible cut is obtained from a tree by pruning some of its sub-trees.
By the reverse process, i.e. by plugging sequences of trees to a fixed tree (that corresponds to an admissible cut),
we obtain the set of trees which admit that admissible cut.
On Figure~\ref{fig:enriched-cut}, the fixed admissible cut is the tree with nodes $a,b,c,d$ and
the places where sequences of trees can be plugged are depicted by the grey triangles.
For every node of arity $\eta$, exactly $\eta +1$ sequences of trees can be plugged near its children.
So for a fixed admissible cut of size $n-i$, the number of places is $2n-2i-1$.
Thus we conclude:
\[W^i_n = [z^n] G_{n-i} z^{n-i} \left( \frac{1}{1-C(z)} \right)^{2n-2i-1} = \frac{(2n-2i-2)!}{2^{n-i-1} (n-i-1)!} [z^{i}] \left(\frac{C(z)}{z}\right)^{2n-2i-1},\]
by using Definition~\ref{def:catalan} and Fact~\ref{fact:number:increasing}.
Results of~\cite[Part 3, Chapter 5]{PS70} on powers of the Catalan generating function give:
\[[z^n] \left(\frac{C(z)}z\right)^k = \frac{k}{n} \binom{k+2n-1}{n-1}.\]
The former result is obtained by using the ``B\"urmann's form'' of Lagrange inversion.
An analogous expression is given in~\cite[p.~66--68]{FS09}.
Thus,
\[W^i_n = \frac{(2n-2i-1)!}{i\cdot 2^{n-i-1}\cdot (n-i-1)!} \binom{2n-2}{i-1}.\]
By taking the average, i.e. by dividing by $C_n$, we obtain the stated value for $\bar{W}^i_n$.
\end{proof}

Given this result, we can complete the analysis of the shapes depicted in Figure~\ref{fig:shape40}.
Let us first determinate the limit curve for the average shape of an semantic tree.
We renormalise the values $\bar{W}^i_n$ as $f(c,n)=\ln(\bar{W}^{\lfloor cn\rfloor}_n)$ 
and we evaluate an asymptotic of $f(c,n)$ when $n$ tends to infinity.
An easy calculation shows that, for $\frac{1}{n}\ll c\ll 1-\frac{1}{n}$ (i.e. $\frac{1}{n}=o(c)$ and $1-c=o(\frac{1}{n}$) 
, we have~:
\[f(c,n)=(1-c)n\ln(n)+\left( c-1+\ln  \left( {\frac { \left( 2-2\,c \right) ^{1-c}}{{c}^{c} \left( 
2-c \right) ^{2-c}}} \right) \right)n+\ln  \left({\frac {\sqrt {4-2c}}{\sqrt {c}}} \right)
+\BigO\left(\dfrac{1}{c(1-c)n}\right).\]
In particular, as $n$ tends to infinity, on every compact $[a,b]$ such $0<a<b<1$,
the function in $c$, $f(c,n)$ tends uniformly to a line $(1-c)n\ln(n)$. Moreover, if we keep the second order terms, 
we then obtain a curve which is totally relevant with the Figure~\ref{fig:shape40}.

Now, we are interested to the behaviour in the neighbourhood of the extremities of $[0,1]$.
We can study the asymptotic of $\bar{W}^i_n$ for fixed constant $i$. 
A straightforward calculation shows that:
\[\bar{W}^i_n\sim_{n\rightarrow\infty} \dfrac{n!}{2^{n-1}i!}, \hspace*{1cm}
\text{and} \hspace*{1cm} \bar{W}^{n-i}_n\sim_{n\rightarrow\infty} \dfrac{(2i-1)!}{2^{k-1}}.\]
Both give an interpretation to the inflexion of the curve near the extremities.

\section{Expected size of process behaviours}\label{sec:size}

 
In this section we study in more details the average size of semantic trees (i.e. the mean number of nodes).
In a first part we provide a first approximation based on the Theorem~\ref{prop:bord} of the previous
section. Then, we make a conjecture regarding the non-plane case, whose proof require a deeper
 study that goes beyond the scope of this paper.
Finally, we characterize the average size in a more precise way, through a linear recurrence that
is obtained by various means. We describe three different
techniques of analytic combinatorics to obtain this recurrence relation: each technique has its
 pros and cons, as will be discussed below.  Finally, we reach our goal of providing 
a precise asymptotics of the size of the semantic trees.

\subsection{First approximation of mean size}

Our initial approximation of the mean size of the semantic trees is based on the level decomposition
 of Section~\ref{sec:shape}, where we give a closed formula for the mean width of nodes $\bar{W}^i_n$ at
 level $n-i-1$ for semantic trees corresponding to syntax trees of size $n$, as Theorem~\ref{prop:bord}.
We first give, as a technical lemma, an inequality involving $\bar{W}^i_n$.

\begin{lemma}\label{lem:bounds}
\[1 \leq \frac{2^{n-1}\ i!}{n!} \cdot \bar{W}^i_n \leq \frac{1}{1-\frac{i^2}{2n}}.\]
\end{lemma}
\begin{proof}
We first deal with some normalized expression of Theorem~\ref{prop:bord}:
\[\frac{2^{n-1}\ i!}{n!} \cdot \bar{W}^i_n = 2^i \cdot \frac{(2n-2i-1)!}{(2n-i-1)!}\cdot \frac{(n-1)!}{(n-i-1)!} =
 \frac{\prod_{j=0}^{i-1} (n-i+j)}{\prod_{j=0}^{i-1} (n-i+\frac{j}{2})}.\]
Obviously, we get the stated lower bound~$1$. Let us go on with the simplification.
The $\lceil i/2\rceil$ first factors of the numerator can be simplified with those of the denominator when $j$ is even:
\[
\frac{2^{n-1}\ i!}{n!} \cdot \bar{W}^i_n =
  \frac{\prod_{j=0}^{\lfloor i/2 \rfloor-1} (n-\lfloor\frac{i}{2}\rfloor+j)}{\prod_{j=0}^{\lfloor i/2 \rfloor-1} (n-i+\frac{1}{2}+j)} =
  \frac{\prod_{j=0}^{\lfloor i/2 \rfloor-1} (1-\frac{\lfloor i/2\rfloor-j}{n})}{\prod_{j=0}^{\lfloor i/2 \rfloor-1} (1-\frac{i-1/2-j}{n})}.
\]
But, the numerator is smaller than 1 and the denominator satisfies:
\[
\prod_{j=0}^{\lfloor i/2 \rfloor -1} (1-\frac{i-1/2-j}{n}) \geq
  1 - \sum_{j=0}^{\lfloor i/2 \rfloor -1} \frac{i-1/2-j}{n} \geq
  1 - \frac{i^2}{2n}.\]
So we obtain the stated upper bound.
\end{proof}

\begin{theorem}\label{theo:average_size}
The mean size $\bar{S}_n$ of a semantic tree induced by a process tree of size $n$ admits the following asymptotics: 
\[\bar{S}_n  = \sum_{i=0}^{n-1} \bar{W}_n^i \sim_{n\rightarrow\infty} e\frac{n!}{2^{n-1}}.\]
\end{theorem} 
\begin{proof}
Using Lemma~\ref{lem:bounds} and taking a large enough $n$, we get:
\[\sum_{i=0}^{n-1} \frac{n!}{2^{n-1}\ i!} \leq \bar{S}_n \leq \sum_{i=0}^{n-1} \frac{n!}{2^{n-1}\ i!} \cdot \frac{1}{1-\frac{i^2}{2n}}
 \leq \sum_{i=0}^{n-1} \frac{n!}{2^{n-1}\ i!} \cdot \left(1 + \frac{i^2}{n}\right).\]
First let us take the lower bound into account. Using an upper bound of the tail of the series (Taylor-Lagrange formula):
\[ \frac{n!}{2^{n-1}} \left( e + \frac{1}{n!} \right) \leq \sum_{i=0}^{n-1} \frac{n!}{2^{n-1}\ i!} \leq \frac{n!}{2^{n-1}} \left( e + \frac{e}{n!} \right),\]
so both bounds tends to $e\cdot n!/2^{n-1}$. It remains to prove
that $n^{-1}\cdot\sum_{i=0}^{n-1} i^2/i! = o(1)$ to complete the proof.
\[\sum_{i=0}^{n-1} \frac{i^2}{i!} = \sum_{i=0}^{n-2} \frac{i+1}{i!} = \sum_{i=0}^{n-3} \frac{1}{i!} + \sum_{i=0}^{n-2} \frac{1}{i!}\rightarrow_{n\rightarrow \infty} 2e.\]
\end{proof}

\begin{corollary}
Let $f$ be a function in $n$ that tends to infinity with $n$. Let $a_n$ be the average number of nodes of the semantic tree induced by all the syntax tree of size $n$ and $l_n$ the average number of nodes belonging to the $f(n)$ last levels.
Then, $l_n/a_n$ tends to 1 when $n$ tends to infinity.
\end{corollary}
The proof is analogous as the previous one using the Taylor-Lagrange formula. 
The unique constraint for $f$ is that it tends to infinity, but it can grow as slow as we want.  For example,
asymptotically almost all nodes of the average semantic tree belong to the $\log (\ldots (\log n) \ldots)$ last levels.  

\subsection{The case of non-plane trees}

In order to compute the average size in the context of non-plane trees, 
we need one more result that is the analogous of powers of the Catalan generating function
(see proof of Theorem~\ref{prop:bord}).
Here, in the case of non-plane trees this corresponds to the powers of
unlabelled non-plane rooted trees.
Although many results about forest of unlabelled non-plane trees have been studied in~\cite{Palmer1979109}, 
it seems that the case of finite sequences of unlabelled non-plane trees has not been thoroughly investigated.

\begin{conjecture}
The mean size $\bar{U}_n$ of a semantic tree induced by a process (unlabelled non-plane rooted)
tree of size $n$ admits the following asymptotics: 
\[\bar{U}_n \sim_{n\rightarrow \infty} \frac{2\sqrt{2} e \pi n}{\gamma}\left(\frac{n\eta}{e}\right)^{n},\]
where $\eta$ and $\gamma$ are introduced in Fact~\ref{fact:catalan_nonplane}.
\end{conjecture}

\subsection{The mean size as a linear recurrence}

In this section, we focus on the asymptotics of the average size $\bar{S}_n$ of the semantic trees induced by syntax trees
of size $n$. Our goal is to obtain more precise approximations than
Theorem~\ref{theo:average_size} using different analytic combinatorics techniques.
Indeed, we present three distinct ways to establish our main result:
a linear recurrence that precisely capture the desired quantity.
These results are deeply related to the holonomy property of the generating functions
into consideration. Thus, a priori, in the non-plane case, the functions are not holonomic
and consequently such proofs could not be adapted.

\begin{theorem}\label{theo:taille}
The mean size $\bar{S}_n$ of a semantic tree induced by a tree of size $n$ follows the P-recurrence: 
\begin{multline*}
(2n^4+12n^3+22n^2+12n)\bar{S}_n-(4n^4+32n^3+87n^2+87n+18)\bar{S}_{n+1}\\
+(2n^4+24n^3+85n^2+106n+39)\bar{S}_{n+2}-(4n^3+20n^2+31n+15)\bar{S}_{n+3}=0,
\end{multline*}
with the initial conditions: $\bar{S}_0 = 0, \bar{S}_1 = 1 \text{ and } \bar{S}_2 = 2$.
\end{theorem}
We have stored the non-normalized version of this sequence in OEIS, at A216234.
It consists to $S_n$: the cumulated sizes of semantic trees issued of process trees of size $n$. 

\begin{proof}
A first approach to prove this theorem is based on creative telescoping.
This proof is a direct consequence of the level decomposition detailed in Section~\ref{sec:shape}.
It is clearly simple both in terms of the technical mathematics involved and the
 level of computer assistance required for the demonstration. In particular, it needs no proof in the sense
 that all the steps are totally automatized in classical computer algebra systems (such a Maple or Mathematica).
The level decomposition is really a peculiarity of the combinatorial structure we investigate, and it is
hardly a common situation.

From the exact formula for the mean number of nodes each level, $\bar{W}^i_n$,
by summing on all levels, we get the mean number $\bar{S}_n$ of nodes on an average semantic tree:
\[\bar{S}_n = \sum\limits_{i=0}^{n-1}\frac{2^i \ (2n-2i-1)!\ (n-1)!}{(2n-i-1)!\ (n-i-1)!} \cdot \frac{n!}{2^{n-1}\ i!}.\]
This sum can be expressed in terms of hypergeometric functions:
\[\bar{S}_n=\frac {n!}{{2}^{n-1}}
{\,{}_1F_1(-2\,n+1;\,-n+1/2;\,1/2)}-
{{}_2F_2(1,-n+1;\,1/2,n+1;\,1/2)}.
\]

Now, by using the package Mgfun of Maple~\cite{chyzak97},
we extract, by creative telescoping~\cite{PWZ96, zeilberger90},
the stated P-recurrence for $\bar{S_n}$.
\end{proof}

However, we can prove Theorem~\ref{theo:taille} with two other distinct approaches.
The first one is based on the multivariate holonomy theory. It is our original proof,
 that can be found in~\cite{BGP12}, and it is clearly
the proof that conveys the most combinatorial informations about the structures we study.\\

Finally, the last approach is based on the concept \emph{guess-and-proof}: we calculate the first values for $S_n$,
guess a differential equation verified by $S(z)$ and prove that it is corrected.  This proof style is both
 powerful and clever since it is almost entirely automated. However, it does not convey much information about the combinatorial structures
 under study. 
%

\subsection{Precise asymptotics of the size}

Now that we have a P-recurrence for the mean size, we can obtain precise asymptotics
 in a relatively effortless way.

\begin{theorem}\label{theo:asympt}
The mean size $\bar{S}_n$ of a semantic tree induced by a tree of
 size $n$ admits the following precise asymptotics: 
\[\bar{S}_n=e\sqrt{2\pi n}\left(\frac{n}{2e}\right)^{n}\left(2+\frac{2}{3n}+
\frac {49}{36n^2}+\frac {27449}{6480n^3}+\BigO \left(\frac{1}{n^4}\right)\right).\]
%
\end{theorem} 

\begin{proof}

We can derive more directly this result from the hypergeometric expression.
 First, let us observe that ${{}_2F_2(1,-n+1;\,1/2,n+1;\,1/2)}$ tends to a constant
 when $n$ tends to infinity. So, we essentially need to analyse the part $\frac {n!}{{2}^{n-1}}
{\,{}_1F_1(-2\,n+1;\,-n+1/2;\,1/2)}$. Let us observe that the next hypergeometric
 function ${\,{}_1F_1(-2\,n+1;\,-n+1/2;\,x)}$ admits the following expansion:
\begin{align*}
 {}_1F_1(-2\,n+1;\,&-n+1/2;\,x)=  e^{2x}+ x^2e^{2x}\dfrac{1}{n}+ \\
 &  x^2e^{2x}\left(\frac{3}{2}+2x+\frac{1}{2}x^2\right)\dfrac{1}{n^2}+
 \frac{1}{12}x^{2}{{\rm e}^{2\,x}} \left( 27+96\,x+84\,{x}^{2}+24\,{x}^{3} +
 2\,{x}^{4} \right)\dfrac{1}{n^3}+\\
& x^2{e}^{2x}\left(\frac {27}{8}+{\frac {49}{2}}{x}+{\frac {349}{8}}{x}^{2}+29{x}^{3}+{\frac {33}{4}}{x}^{4}+{x}^{5}+\frac{1}{24}{x}^{6}\right)\dfrac{1}{n^4}+
\BigO\left(\dfrac{1}{n^5}\right).
\end{align*}
Thus the asymptotic of $\bar{S}_n$ follows directly. 

Let us remark that it also possible to reach this asymptotic from the P-recurrence. We introduce a new auxiliary generating function which is more tractable than $S(z)$. 
For that purpose, recall that the total number of leaves in the semantic trees
induced by process trees of size $n$ (which is also the number of increasing trees of size $n$)
is equal to $n!/2^{n-1}$. 
So, it is natural to study the series $R(z)$ with general terms $\bar{S}_n2^{n-1}/n!$
which is also holonomic and verifies:
\begin{multline*}
-2\left( 10{z}^{2}+7z+3 \right) R(z) +
 \left( -16{z}^{4}+32{z}^{3}+18{z}^{2}+7z \right) 
 R'(z) +\\ 
4\left( 4z^{4}-6{z}^{3}-{z}^{2} \right) 
 R''(z) +
4\left( -{z}^{4}+{z}^{3} \right) 
 R'''(z) =4{z}^{2}+z,
\end{multline*}
with the initial conditions $R(0)=0, R'(0)=1$, $R''(0)=4$. 
The coefficients $R_n$ follow the P-recurrence:
\begin{multline*}
-16nR_n+4(4n^2+12n+3)R_{n+1}-2(2n^3+18n^2+31n+13)R_{n+2}+\\
(4n^3+20n^2+31n+15)R_{n+3}=0,
\end{multline*}
with $R_0 = 0, R_1 = 1$ and $R_2 = 2$.
Now, we can easily prove that this recurrence is convergent.
Indeed, the recurrence is non-negative and asymptotically decreasing,
just by observing that $R_{n+3}-R_{n+2}=\left(\frac{4}{n}+\BigO(\frac{1}{n^2})\right)\left(R_{n+2}-
R_{n+1}\right)+\BigO\left(\frac{1}{n^2}\right)R_{n}$
implies that for $n$ sufficiently large the difference is always negative.

Theorem \ref{theo:average_size}  shows that the series converge to $\exp(1)$.
Now, a deeper analysis of this recurrence can be done using the tools described in \cite[p.~519--522]{FS09}.
 Indeed, the singularities are regular.

 Another way consists in predicting that the asymptotic
 expansion of $R(z)$ as $z$ tends to the infinity can be expressed as 
$\exp (2z+a\ln  \left( z \right) +b{z}^{-1}+c{z}^{-2}+d{z}^{-3}+\BigO(z^{-4}))$
 and to use saddle point analysis (its hypotheses being validated by Wasow's theory) to conclude.


\end{proof}

\section{Applications} \label{sec:randgen}

We describe in this section two practical outcomes of our quantitative study of the pure
merge operator. First, we present an algorithm to efficiently compute the uniform probability
 of a concurrent run prefix.  The second application is a
uniform random sampler of concurrent runs.  These algorithms work directly on the syntax
trees without requiring the explicit construction of the semantic trees. An important remark
 is that these algorithms continue to apply whether we consider the plane or the non-plane case,
 only the average quantities are impacted.

\subsection{Probability of a run prefix}

We first describe an algorithm to determinate the probability of a concurrent run prefix
(i.e. the prefix of a branch in a semantic tree). In practice, this algorithm can be used
to guide a search in the state space of process behaviours, e.g. for statistical model checking or
(uniform) random testing.

\begin{definition}
Let $T$ be a process tree and $\sigma=\langle \alpha_1,\ldots,\alpha_p\rangle$ a prefix
of a run in $\Shuf(T)$. The \textbf{suspended tree} $T_\sigma$ has root $\alpha_p$
and its children are all children of the nodes $\alpha_1,\dots, \alpha_p$ not already in $\sigma$ and
 ordered according to the prefix traversal of $T$.
\end{definition}
For example, the suspended tree  $T_{\langle a,b,d \rangle}$  of the syntax tree $T$ of Figure~\ref{fig:task-and-shuffle} has root $d$ and children (from left to right): the leaves $c,e$ and $f$.\\

Let $T$ be a process tree and $\sigma_{p} = \langle \alpha_1,\ldots,\alpha_p\rangle$ a run prefix
of length $p$.  We are interested in the probability of choosing $\alpha_{p+1}$ to form the prefix run $\sigma_{p+1} = \langle \alpha_1,\ldots,\alpha_p, \alpha_{p+1}\rangle$.
To obtain this probability, we need to count how many runs in $T_{\sigma_{p}}$ are 
first running $\alpha_{p+1}$.
\begin{proposition} \label{prop:proba_pref}
\[\IP(\sigma_{p+1} \mid \sigma_p) = \frac{\ell_{T_{\sigma_{p+1}}}}{\ell_{T_{\sigma_p}}}
 = \frac{(|T|-p) !}{\prod_{S \text{ sub-tree of }T_{\sigma_{p+1}} } |S|}
  \cdot  \frac{\prod_{S \text{ sub-tree of }T_{\sigma_{p+1}} } |S|}{(|T|-p+1) !}
 = \frac{|T(\alpha_{p+1})|}{|T|-p+1}. \]
\end{proposition}
The proof directly derives from the hook-length formula (cf. Fact~\ref{fact:hook}).

\begin{corollary}\label{cor:proba_pref}
Let $T$ be a process tree and $\sigma=\langle \alpha_1,\ldots,\alpha_p\rangle$ be a prefix
of a run in $\Shuf(T)$. For the uniform probability distribution on the set of all
concurrent runs,
the induced probability on prefixes satisfies $\IP(\sigma) = \ell_{T_{\sigma}} / \ell_T$.
\end{corollary}

\begin{corollary}\label{cor:proba_pref2}
Let $T$ be a process tree. The probability of a prefix run $\sigma=\langle \alpha_1,\ldots,\alpha_p\rangle$
in the shuffle tree of $T$ is equal to $\prod^p_{k=1} |T(\alpha_{k})| / (|T|-k+1)$.
\end{corollary}

\begin{algorithm}
\KwData{$T$: a weighted process tree of size $n$}
\KwData{$\sigma:=\langle \alpha_1,\ldots,\alpha_p\rangle$: a run prefix of length $p\leq n$}
\KwResult{$\rho_{\sigma}$: the probability of $\sigma$ in the shuffle of $T$}

$\rho_\sigma := 1$\\
$i:=1$\\

\For {i from 1 to $p-1$}{
       $\rho_{\sigma} := \rho_{\sigma} \times \frac{|T(\alpha_{i+1})|}{n-i}$ \\
       $i := i+1$}
\Return $\rho_{\sigma}$
\caption{\label{algo:prefix:proba} probability of a run prefix.}
\end{algorithm}

From Corollaries~\ref{cor:proba_pref} and~\ref{cor:proba_pref2}  we derive
as Algorithm~\ref{algo:prefix:proba} the computation of the probability~$\rho_{\sigma}$ of
 a concurrent run prefix $\sigma$. While measuring the probability in terms of a semantic tree, the latter need
 not be constructed explicitly. The algorithm indeed requires only the syntax tree $T$
 with added weights, and a few extra memory cells. It trivially performs
 in linear-time.

\begin{proposition} \label{prop:prefix:complexity}
Algorithm~\ref{algo:prefix:proba} computes $\rho_\sigma$ in $p-1$ steps, and $\Theta(p)$ arithmetic operations.
\end{proposition}

If the run prefix $\sigma$ we consider is a full run (i.e. a complete traversal of a syntax tree~$T$), then
 we obtain the uniform probability of a run in general  (since all runs have equal probability in
the semantic tree). Then we have the following result.

\begin{proposition}
The number of concurrent runs of a process tree $T$ is $1/\rho_{\sigma}$ when $\sigma$ is any complete traversal of $T$.
\end{proposition}

As a matter of fact, Algorithm~\ref{algo:prefix:proba} provides us \emph{for free} a
 way to compute from a syntax tree $T$ the number of concurrent runs $\ell_T$ in
the corresponding semantic tree. For this we simply have to compute the probability
 of a full run $\sigma$  (we might select an arbitrary traversal of $T$) and then we
obtain $\ell_T = 1/\rho_{\sigma}$. 

From an order-theoretic point of view we thus we obtain as a by-product a linear-time algorithm
 to compute the number of linear extensions of a tree-like partial order.

\begin{corollary}
Let $T$ a tree-like partial order of size $n$. The number of its linear extensions can be computed in $\BigO(n)$. 
\end{corollary}

Since any full run has length the size $n$ of the syntax tree, the upper-bound $\BigO(n)$ is trivially obtained. Moreover,
 we conjecture that the problem has $\Omega(n)$ lower-bound also.  Note that the hook length formula also yields
 a linear-time algorithm but with more arithmetic operations. To put into
a broader perspective this result, we remind the reader that the problem of counting linear extensions of
 partial orders is $\#P$-complete~\cite{BW91}  in the general case. Moreover, the proposed solution
 (obtained thanks to the very fruitful isomorphism with increasing trees)  is clearly an improvement if
compared to the quadratic algorithm proposed in~\cite{Atkinson90}.

\subsection{Random generation of concurrent runs}

The uniform random generation of concurrent runs is of great practical interest. The problem has a trivial solution if we work
 on a semantic tree. Since all runs have equal probability, we may simply select a leaf at random, and reconstruct the full run
by climbing the unique branch from the selected leaf to the root of the tree. Of course, this naive algorithm is highly impractical
given the exponential size of the semantic tree. The challenge, thus, is to find a solution which does not require the explicit construction
of the semantic tree.   A possible way would be to rely on a Markov Chain Monte Carlo (MCMC) approach, e.g. based on~\cite{DBLP:journals/dm/Huber06}. We describe here a simpler, more direct approach that yields a more efficient sub-quadratic algorithm.

The main idea is to sample in a multiset containing the nodes of the syntax trees as elements,
 each one associated
 to a weight corresponding to the size of the sub-tree rooted at this node.
 A particularly efficient way to implement the required multiset structure is to use a \emph{partial sum tree}, i.e. a balanced binary
 search tree in which all operations (adding or removing a node)
are done in logarithmic time. The details of this implementation can be found in Appendix~\ref{sec:randmset}.

\begin{algorithm}
\KwData{$T$: a weighted process tree of size $n$}
\KwResult{$\sigma$: a run (a list of nodes)}

$\sigma := \langle \rangle$ \\
$M := \{\!\!\{ a^{|T|} \}\!\!\}$ \hspace{1.2cm}\# initialize a multiset with the root $a$ with its weight\\

\For {p \emph{\textbf{from}} $1$ \emph{\textbf{to}} $|T|-1$}{
	$\alpha := \fun{sample}(M)$ \hspace{0.2cm} \# sample an action $\alpha$ according to its weight in the multiset\\
 	$\sigma := \sigma . \alpha$ \hspace{1.05cm}\# append the sampled action to the sequence\\  
        $M := \fun{update}(M,\alpha,0)$ \hspace{1.7cm}\# $\alpha$ cannot be sampled anymore\\
        \For{$\beta\in \fun{child}(T_\sigma)$} {
          $M := \fun{update}(M,\beta,|T(\beta)|)$ \hspace{1cm}\# insert the children of $\alpha$ in the multiset\\ }
}	
\Return $\sigma$
\caption{\label{algo:randgen} uniform random generation of concurrent runs}
\end{algorithm}

Let $T$ be a process tree. First by one traversal, we add a label to all nodes of $T$ that
corresponds to the size of the sub-tree rooted in that node. We say that this size corresponds
to the weight of each node.
We build a list $\sigma$, at the end of size $n$,
such that at each step $i$, we add one action to $\sigma$ that corresponds to the $i$-th 
action in our random run. To choose this $i$-th action, we sample in the multiset of all
possible actions available in the considered step. Initially only the root is available (with probability
$1$ thus cardinality $n$ the size of the process tree $T$). Then it is added to $\sigma$ and removed
 from the multiset. Finally its children are enabled with the weight as cardinality. And we proceed until
all actions have been sampled. 

Let $T$ a syntax tree, we denote by $\fun{child}(T)$ the nodes at level one of $T$.

The following loop invariant derives easily from Algorithm~\ref{algo:randgen}.

\begin{invariant} \label{invariant:randgen}
At the $p$-th step of the algorithm, we have:
 \[|M_p| = |T| - p + 1 \text{ and } \overline{M_p}=\{\alpha_{p+1} \mid \alpha_{p+1} \in \fun{child}(T_{\sigma_p})\}.\]
\end{invariant}

\begin{proposition}
Let $\sigma_p$ the prefix obtained at the $p$-th step  in algorithm~\ref{algo:randgen}.
The next action $\alpha_{p+1}$ is chosen with probability $|T(\alpha_{p+1})| / (|T|-p+1)$.
Consequently the complete run $\sigma$ is generated with uniform probability.
\end{proposition}

\begin{proof}
Let $M_p$ the multiset obtained at step $p$ in algorithm~\ref{algo:randgen}.
We select the next action $\alpha_{p+1}$ with probability $\frac{M_p(\alpha_{p+1})}{|M_p|}$
(cf. Appendix~\ref{sec:randmset} for a detailed proof).
By Invariant~\ref{invariant:randgen} we have $|M_p|=|T|-p+1$. Moreover, in the algorithm
 we insert $\alpha_{p+1}$ with weight $M_p(\alpha_{p+1}) = |T(\alpha_{p+1})|$. Thus by Proposition~\ref{prop:proba_pref} the prefix $\sigma_{p+1}$ is obtained with the correct probability so that when completed the
 full run $\sigma$ is generated with uniform probability.
\end{proof}

In the case of the partial sum tree implementation, we have the following
complexity results.

\begin{proposition}
Let $n$ be the size of the weighted process tree $T$. To obtain
a random execution, we need $n$ random choices of integers and the operations on the multiset
are of order $\Theta(n\log n)$ (for the worst case).
\end{proposition}

\section{Conclusion and perspectives}\label{sec:conclusion}

The quantitative study of the pure merge operator represents a preliminary step towards our goal of investigating concurrency theory
from the analysis of algorithms point of view. In the next step, we shall address other typical constructs of formalisms for concurrency, especially \emph{non-deterministic choice} and \emph{synchronization}~\cite{DBLP:conf/concur/AcetoI07}.
There are indeed various forms of synchronization, in general corresponding to reflecting the action labels within the pure merge operator. Other operators, such as \emph{hiding}, also deserve further investigations.
We also wish to further investigate the case of non-plane process trees. Although the nature of the operators does not seem to be really impacted (confirming the intuition of Flajolet and Sedgewick), the technical aspects in terms of analytic combinatorics are quite interesting. Another interesting continuation of the work would be to study the compaction of the semantic trees by identifying common sub-trees. This would amount
 to study the interleaving of process trees up-to \emph{bisimilarity}, the natural notion of equivalence for concurrent processes. Note that our algorithmic framework would not be affected by such studies, since they do not require the explicit construction of the semantic trees (whether compacted or not, plane or non-plane).

Perhaps the most significant outcome of our study is the emergence of a deep connection between concurrent processes and increasing labelling of combinatorial structures. We indeed connected the pure merge operator with increasing trees to measure the number of concurrent runs. We also define the notion of increasing admissible cut to study the number of nodes by level in the semantic trees. We expect the discovery of similar increasingly labelled structures while we go deeper into concurrency theory. 

 From a broader perspective, we definitely see an interest in reinterpreting \emph{semantic objects} (from logic, programming language theory, concurrency theory, etc.)
under the lights of analytic combinatorics tools. Such objects (like semantic trees) may be quite intricate
 when considered as combinatorial classes, thus requiring non-trivial techniques. This is highlighted here e.g. by the generalized hook-length formula characterizing the
expected size of semantic trees. Conversely, we think it is interesting to know -- precisely, not just by intuition -- the high-level of sharing and symmetry within semantic trees. This naturally leads to
practical algorithms, making us confident that real-world applications (in our case, especially related to random testing and statistical model-checking) might result from such study.

\noindent \textbf{Acknowledgements.} We are grateful to M. Dien and O. Roussel for fruitful remarks about the algorithms.

\bibliographystyle{alpha}
\bibliography{shuffle}

\newpage
\appendices

\section{Weighted random sampling in dynamic multisets}\label{sec:randmset}

This appendix discusses the problem of random sampling elements according to their respective weight in a multiset. Moreover, the
 multiset must be dynamic in that the cardinality of elements can be changed on-the-fly. This problem represents
a basic algorithmic component  in the random generation of concurrent runs (cf. Section~\ref{sec:randgen}). To our knowledge, this has not been addressed precisely in the literature (some basic information can be found in~\cite{WE80, OR95}). 

\subsection{Dynamic multiset basics}

In this section we recall a few concepts and basic notations of multisets,
the reader may consult e.g.~\cite{DBLP:journals/ndjfl/Blizard89} for a more thorough treatment.
A finite multiset (or bag) $M$ can be defined formally as a function from a \emph{carrier set}
 $\overline{M}$ to positive integers, more precisely 
an injective function $M \subset \overline{M} \rightarrow \mathbb{N}$. 
As an example we consider a multiset $M_0= \{\!\!\{ a^2,b^3,c^1\}\!\!\}$ with
carrier set $\overline{M_0}=\{a,b,c\}$. In the common functional notation,
we would denote $M_0=\{(a,2),(b,3),(c,1) \}$. Each positive integer associated
 to an element is called its \emph{weight} (i.e. number of ``occurrences'') in the multiset.
 The weight of and element $\alpha \in \overline{M}$ is denoted $M(\alpha)$.
 And by a slight abuse of notation we write $\alpha \in M$ \emph{iff} $M(\alpha)\geq 1$.
For example, the element $a$ has weight $2$in $M_0$, thus $M_0(a)=2$ and of course $a\in M_0$.
The notation $\alpha \notin M$ may either denote $\alpha \notin \overline{M}$ or $M(\alpha)=0$.
 This slight ambiguity has interesting algorithmic implications.

The cardinal or \emph{total weight} of $M$ is $|M|=\sum_{\alpha\in M} M(\alpha)$, 
for example $|M_0| = 2+3+1 = 6$ whereas for the carrier set we have $|\overline{M_0}|=3$.

Given a multiset $M$ the two operations we are interested in are:
\begin{itemize}
\item a random sampler $\mathsf{sample}(M)$ that generates an element
 $\alpha\in M$ at random with probability $\frac{M(\alpha)}{|M|}$.
\item an update operation $\mathsf{update}(M,\alpha,k)$ that produces a multiset
 $M'$ such that $\forall \beta\in M, \beta\neq \alpha \implies M'(\beta)=M(\beta)$ and $M'(\alpha)=k$.
\end{itemize}

Remark that if $M'=\mathsf{update}(M,\alpha,0)$ we do have $M'(\alpha)=0$
 hence $\alpha\notin M'$ but it is left unspecified whether $\alpha \in \overline{M'}$ or not.

\subsection{A naive random sampler}

Probably the fastest way to implement the $\mathsf{sample}$ operation 
is to represent a multiset $M$ with an array of length $n=|M|$.
Formally, this defines a finite sequence, i.e. a function $\sigma_M:[1..n] \rightarrow |M|$.
Consider $M=\{\!\!\{ a^{k_1}_1, \ldots, a^{k_n}_n\}\!\!\}$ an arbitrary multiset.
 The cardinality of $M$ is $|M|=\sum^n_{j=1} k_j$ and its carrier set is $\overline{M}=\{\alpha_1,\ldots,\alpha_n\}$.
For the sake of simplicity and without any loss of generality, 
we assume an implicit strict ordering $\alpha_1 < \ldots < \alpha_n$.
Now we define $\sigma_M$ such that
$\forall i\in [1..n],~\forall j \in \left [  \sum^{i-1}_{p=1} k_p + 1 .. \sum^{i-1}_{p=1} k_p + k_i  \right ],~\sigma_M(j)=\alpha_i$, and  everywhere else $\sigma_M$ is undefined.
  For example $\sigma_{M_0}$ is $\{ 1 \mapsto a, 2 \mapsto a, 3 \mapsto b, 4 \mapsto b, 5 \mapsto b, 6 \mapsto c \}$
 which we may also denote $\langle a,a,b,b,b,c \rangle$.

\begin{algorithm}
  \KwData{$M=\{\!\!\{ a^{k_1}_1, \ldots, a^{k_n}_n\}\!\!\}$}
  \KwResult{$\beta$ an element of $M$ taken with probability $\frac{M(\beta)}{|M|}$}
   \BlankLine
  $\rho := $ a uniform random integer taken in range $[1..n]$\; 
  \Return $\beta=\sigma_M(\rho)$\;
\caption{\label{algo:randmset:naive} naive random sampler for dynamic multisets.}
\end{algorithm}

The naive random sampler is described by Algorithm~\ref{algo:randmset:naive}.
For a multiset $M$, we first pick a uniform random integer in range $[1..|M|]$.
This way, we select the $\rho$-th position in the sequence $\sigma_M$ with probability $\frac{1}{|M|}$. 
Now let $i$ such that $\beta=\sigma_M(\rho)=\alpha_i\in M$.
 By definition of $\sigma_M$ we have
 $|\{j \mid \sigma_M(j)=\beta\}|=\sum^{i-1}_{p=1} k_p + k_i - \sum^{i-1}_{p=1} k_p -1 + 1 = k_i = M(\alpha_i)$.
 Thus the probability of picking element
 $\beta=\sigma_M(\rho)$ is $\frac{1}{|M|}\times M(\beta)=\frac{M(\beta)}{|M|}$ which is as required.

The complexity of the sampling algorithm corresponds to the uniform random
 sampling of an integer in range $|1..n|$ for a multiset of total weight $n$,
 plus a single access to the array $\sigma_M$ which is in general performed in constant time. 
 The space complexity is linear in $n$ since we must record $\sigma_M$ with its $|M|$ elements,
 which is not very good since the weight of a given element can be arbitrarily high. 
  Moreover, the update operation is not very efficient for the same reason.

\subsection{A more efficient random sampler based on partial sum trees}

We now describe a random sampler that has far better space requirements -- in the order of $|\overline{M}|$ --
and also enjoys a much more efficient update operation.
 The main idea is to exploit a representation based on \emph{partial sum trees}~\cite{pst:dehne:89}.

\begin{figure}
\begin{center}
\begin{tikzpicture}[node distance=40pt]
\node[rectangle,draw] (a) {$9 \mid a^8 \mid 17$};
\node[below of=a] (ha) {};
\node[rectangle,draw, left of =ha] (b) {$4 \mid b^4 \mid 1$}; 
\node[below of=b] (hb) {};
\node[rectangle,draw, left of =hb, node distance=20pt] (d) {$d^4$}; 
\node[rectangle,draw, right of =hb, node distance=20pt] (f) {$f^1$}; 

\node[rectangle,draw, right of =ha] (c) {$8 \mid c^9 \mid 0$};
\node[below of=c] (hc) {};
\node[rectangle,draw, left of =hc, node distance=20pt] (e) {$e^8$}; 

\draw (a) -- (b);
\draw (a) -- (c);
\draw (b) -- (d);
\draw (b) -- (f);
\draw (c) -- (e);

\end{tikzpicture}
\end{center}
\caption{\label{fig:partial:sum:tree}A partial sum tree for multiset $M_2=\{\!\!\{ a^8,b^4,c^9,d^4,f^1,e^8 \}\!\!\}$.}
\end{figure}
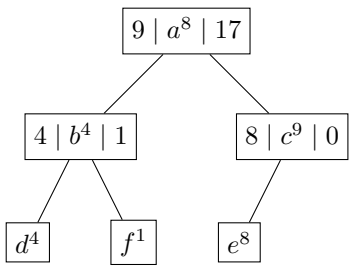

In Figure~\ref{fig:partial:sum:tree} we give a possible partial sum tree (PST) representation of the multiset $M_2=\{\!\!\{ a^8,b^4,c^9,d^4,f^1,e^8 \}\!\!\}$.
The idea is to represent a multiset $M$ as a binary tree with nodes labelled with three informations: 
the total weight of the left and right sub-trees as well as
a unique element $\alpha$ of $M$ together with its weight $M(\alpha)$.

\SetKwFunction{DispatchFun}{dispatch}
\SetKwBlock{Indent}{}{}
\begin{algorithm}[H]
\KwData{$T_M$ a PST for a multiset $M=\{\!\!\{ a^{k_1}_1, \ldots, a^{k_n}_n\}\!\!\}$}
\KwResult{$\beta$ an element of $M$ taken with probability $\frac{M(\beta)}{|M|}$}
\BlankLine
$\rho := $ a uniform random integer taken in range $[1..|M|]$\;  
\Return $\beta = $ \DispatchFun{$T_M$, $\rho$}\;
\BlankLine
where:
\BlankLine
\DispatchFun{$\left [ L \mid \alpha^k \mid R \right ]$, $\rho$} is {
\Indent{
  \uIf{$\rho \leq |L|$}{\Return \DispatchFun{$T_L$, $\rho$}}
  \uElseIf{$\rho  - |L| \leq k$}{\Return $\alpha$}
  \uElse{\Return \DispatchFun{$T_R$, $\rho - (|L| + k)$}}
}
}

\caption{\label{algo:randmset:pst} partial sum tree (PST) based random sampler for dynamic multisets.}
\end{algorithm}

The random sampler based on the PST representation is described by Algorithm~\ref{algo:randmset:pst}.
The first step is the same as in the naive algorithm: pick an integer $\rho$ uniformly at random in the range $[1..|M|]$.
The second part is a simple recursive dispatch within the tree $T_M$ depending only on the value of $\rho$.
If $\rho$ is less than the total weight $|L|$ of the left-sub-tree, denoted $T_L$, of $T_M$ then we pick the element in this left-sub-tree.
If otherwise $\rho$ is in the range $[|L|..|L|+k]$ then we pick-up the root element $\alpha$. Otherwise, we pick the
element in the right-sub-tree $T_R$ without forgetting to update $\rho$ as $\rho - (|L| + k)$ in the recursive calls.

\begin{proposition}
If we assume the tree $T_M$ to be well-balanced, the worst-case time complexity of the PST-based random sampler is $\BigO(\log |\overline{M}|)$.
The update operation inherits the same worst-case complexity.
 Moreover the PST itself occupies space of order $\Theta(|\overline{M}|)$ in memory.
\end{proposition}

The well-balanced assumption is easy to obtain in practice, either by relying on implicitly well-balanced tree models e.g. AVL or red-black trees, or by
simply constructing the PST in a deterministically well-balanced way (e.g. with a bit flag in each node, flipped after each insertion). 
So if compared to the naive algorithm and its constant-time sampling, the PST algorithm is far better in terms of memory usage. The most prominent
advantage is a now very efficient update operation: we just need to update the left and right sums in the nodes from the updated node to the root of the tree
(trivially also in order $\BigO(\log n)$ for a well-balanced tree).

The proof for the correctness property is now slightly more involved.

\begin{proposition}
Let $M$ a multiset. The PST random sampler returns an element $\alpha \in M$ with probability $\frac{M(\alpha)}{|M|}$. 
\end{proposition}

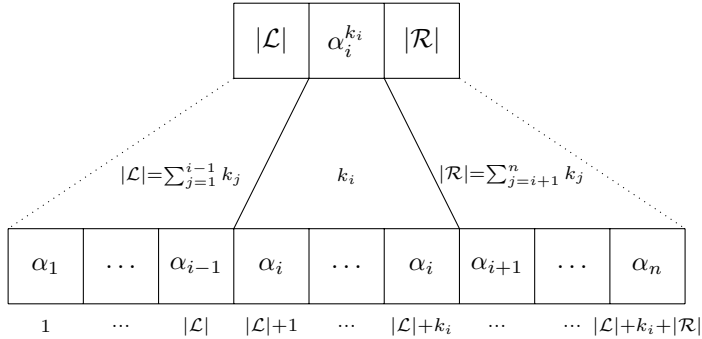
\begin{figure}
\noindent \begin{centering}
\begin{minipage}[t]{0.7\columnwidth}%
\begin{tikzpicture}
\draw (0,0) grid (3,1);
\draw (0.5,0.5) node {$|\mathcal{L}|$};
\draw (1.5,0.5) node {$\alpha^{k_i}_i$};
\draw (2.5,0.5) node {$|\mathcal{R}|$};
\draw (-2.5,-2.5) node {$\alpha_1$};
\draw (-1.5,-2.5) node {$\ldots$};
\draw (-0.5,-2.5) node {$\alpha_{i-1}$};
\draw (0.5,-2.5) node {$\alpha_i$};
\draw (1.5,-2.5) node {$\ldots$};
\draw (2.5,-2.5) node {$\alpha_i$};
\draw (3.5,-2.5) node {$\alpha_{i+1}$};
\draw (4.5,-2.5) node {$\ldots$};
\draw (5.5,-2.5) node {$\alpha_{n}$};
\draw[dotted] (0,0) -- (-3,-2);
\draw (1,0) -- (0,-2);
\draw (2,0) -- (3,-2);
\draw[dotted] (3,0) -- (6,-2);
\draw (-2.5,-3.3) node {${\scriptstyle 1}$};
\draw (-1.5,-3.3) node {${\scriptstyle \ldots}$};
\draw (-0.5,-3.3) node {${\scriptstyle |\mathcal{L}|}$};
\draw (0.5,-3.3) node {${\scriptstyle |\mathcal{L}|+1}$};
\draw (1.5,-3.3) node {${\scriptstyle \ldots}$};
\draw (2.5,-3.3) node {${\scriptstyle |\mathcal{L}|+k_i}$};
\draw (3.5,-3.3) node {${\scriptstyle \ldots}$};
\draw (4.5,-3.3) node {${\scriptstyle \ldots}$};
\draw (5.5,-3.3) node {${\scriptstyle |\mathcal{L}|+k_i+|\mathcal{R}|}$};
\draw (-0.7,-1.3) node {${\scriptstyle |\mathcal{L}|=\sum^{i-1}_{j=1}k_j}$}; 
\draw (1.5,-1.3) node {${\scriptstyle k_i}$};
\draw (3.7,-1.3) node {${\scriptstyle |\mathcal{R}|=\sum^{n}_{j=i+1}k_j}$};
\draw (-3,-3) grid (6,-2);
\end{tikzpicture}%
\end{minipage}
\par\end{centering}

\caption{\label{fig:sigma:mapping}Mapping from a node of the partial sum tree
to its corresponding infix sequence.}
\end{figure}

\begin{proof}
Let $M=\{\!\!\{\alpha_{1}^{k_{1}},\ldots,\alpha_{n}^{n}\}\!\!\}$
a multiset and let denote by $T_M$ the its partial sum
tree representation. If we assume an implicit ordering $\alpha_{1}<\ldots<\alpha_{i}<\ldots<\alpha_{n}$
then the infix sequence%
\footnote{The infix sequence of a binary tree is the sequence of its elements
ordered according to the infix traversal of the tree.}
of $M$ is exactly $\sigma_{M}$ as defined previously.
More generally, a given node $N=(L,\alpha_{i}^{k_{i}},R)$ of the representation corresponds
to a multiset $N$, and we denote $\sigma_{N}$
its infix sequence. The tree $T_L$ is the left sub-tree corresponding
to a multiset $L$ with infix sequence $\sigma_{L}$
of elements from $\sigma_{N}(1)$ to $\sigma_{N}(\sum_{j=1}^{i-1}k_{j})$.
The root node of $T_N$ represents the element $\alpha_{i}$ with cardinality $k_{i}$
which maps to the sub-sequence $\sigma_{\alpha_{i}}$ of elements
from $\sigma_{N}(\sum_{j=1}^{i-1}k_{j}+1)$ to $\sigma_{N}(\sum_{j=1}^{i}k_{j})$.
Finally $T_R$ is the right sub-tree corresponding to multiset $R$
with infix sequence $\sigma_{R}$ of elements from $\sigma_{N}(\sum_{j=1}^{i}k_{j}+1)$
to $\sigma_{N}(\sum_{j=1}^{n}k_{j})$. This mapping is depicted
on Figure~\ref{fig:sigma:mapping}. Thus the left sub-tree $T_L$ represents
the multiset $L=\{\!\!\{\alpha_{1}^{k_{1}},\ldots,\alpha_{i-1}^{k_{i-1}}\}\!\!\}$
and $T_R$ represents the multiset $R=\{\!\!\{\alpha_{i+1}^{k_{i+1}},\ldots,\alpha_{n}^{k_{n}}\}\!\!\}$.

Now we demonstrate the property that at node $N=(L,\alpha_{i}^{k_{i}},R)$
the call dispatch($T_N$,$\rho$) with $\rho$ taken randomly in $[1..|N|]$
yields an element $\alpha\in N$ with probability $\frac{N(\alpha)}{|N|}$.
We proceed by induction on the tree structure (or size). Suppose $\rho\in[1..|L|]$
then :

\begin{itemize}
\item if $T_N$ is a leaf $(\emptyset,\alpha^{k},\emptyset)$ then $\rho\in[1..k]$
with $k=N(\alpha)$ and $\sigma_{N}=[\alpha,\ldots,\alpha]$
(i.e. $ran(\sigma)=\left\{ \alpha\right\} )$ . Then the probability
of choosing $\alpha$ is $\frac{N(\alpha)}{|N|}=\frac{k}{k}=1$.
\item if $T_N=(L,\alpha_{i}^{k_{i}},R)$ is an internal node, i.e. $L\cup R\neq\emptyset$
then we show the property to hold for $T_N$ if we assume it holds for
the sub-trees $T_L$ and $T_R$ :

\begin{itemize}
\item if $1\leq\rho\leq|L|$ then dispatch($T_N$,$\rho$)=dispatch($T_L$,$\rho$)
and the property holds by hypothesis of induction.
\item $ $if $|L|+1\leq\rho\leq|L|+k_{i}$ then we select
element $\alpha_{i}$ with probability $\frac{k_{i}}{|L|+k_{i}+|R|}=\frac{N(\alpha_{i})}{|N|}$
\item if $|L|+k_{i}\leq\rho\leq|N|$ then dispatch($T_N$,$\rho$)=dispatch($T_R$,$\rho'$)
with $\rho'=\rho-(|L|+k_{i})$ and since $\rho'\in[1..|R|]$
the property holds by hypothesis of induction.\end{itemize}
\end{itemize}

In consequence, if $\rho\in[1..|M|]$ then dispatch($T_M$,$\rho$)
yields an element $\alpha$ with probability $\frac{M(\alpha)}{|M|}$.
\end{proof}

Note that the proof exploits the fact that the mapping from a multiset
$M$ to its representative sequence $\sigma_{M}$
is deterministic. While imposing a total order on the elements $\alpha_{i}$'s
provides a simpler solution -- the representative being the infix
- ordered sequence -- it is by no means a strict requirement. The
proof can easily - albeit uninterestingly - adapt to any permutation
in the selected order.

\label{lastpage}

\end{document}